\documentclass{amsart}[11pt]
\usepackage[margin=3cm]{geometry}
\usepackage{amsmath,amsfonts, amsthm, amssymb,mathtools}
\usepackage{hyperref}
\usepackage{enumerate}
\usepackage{color}
\usepackage{float}
\usepackage{subfigure}

\usepackage[colorinlistoftodos,prependcaption,textsize=tiny]{todonotes}
\presetkeys{todonotes}{inline}{}

\numberwithin{equation}{section}
\newtheorem{theorem}{Theorem}[section]
\newtheorem{corollary}[theorem]{Corollary}
\newtheorem{lemma}[theorem]{Lemma}
\newtheorem{proposition}[theorem]{Proposition}

\theoremstyle{definition}

\newtheorem{example}[theorem]{Example}
\newtheorem{algo}[theorem]{Algorithm}
\newtheorem{remark}[theorem]{Remark}
\newtheorem{assumption}[theorem]{Assumption}

\newcommand{\QQ}{\mathbb{Q}}
\newcommand{\RR}{\mathbb{R}}
\newcommand{\PP}{\mathbb{P}}
\newcommand{\VV}{\mathbb{V}}

\newcommand{\D}{\mathrm{d}}
\newcommand{\E}{\mathrm{e}}
\newcommand{\EE}{\mathbb{E}}
\newcommand{\Ff}{\mathcal{F}}
\newcommand{\Oo}{\mathcal{O}}
\newcommand{\Tt}{\mathcal{T}}
\newcommand{\eps}{\varepsilon}
\newcommand{\half}{\frac{1}{2}}
\newcommand{\Hm}{H_{-}}
\newcommand{\Hp}{H_{+}}
\newcommand{\XiT}{\Xi_{T}}

\newcommand{\Ptilde}{\widetilde{P}}
\newcommand{\Cf}{\mathfrak{C}}
\newcommand{\Df}{\mathfrak{D}}
\newcommand{\Wf}{\mathfrak{W}}
\newcommand{\tf}{\mathfrak{t}}
\DeclareMathOperator{\sgn}{sgn}

\def\blue#1{\textcolor{black}{#1}}

\begin{document}
\title{Interest rate convexity in a Gaussian framework}
\date{\today}
\author{Antoine Jacquier}
\address{Department of Mathematics, Imperial College London, and the Alan Turing Institute}
\email{a.jacquier@imperial.ac.uk}
\author{Mugad Oumgari}
\address{University College London and Lloyds Banking}
\email{Mugad.Oumgari@lloydsbanking.com}
\thanks{The authors would like to thank Damiano Brigo for helpful comments.
AJ is supported by the EPSRC grants EP/W032643/1 and  EP/T032146/1.
‘For the purpose of open access, the author(s) has applied a Creative Commons Attribution (CC BY) licence (where permitted by UKRI, ‘Open Government Licence’ or ‘Creative Commons Attribution No-derivatives (CC BY-ND) licence’ may be stated instead) to any Author Accepted Manuscript version arising’.}

\maketitle

\begin{abstract}
The contributions of this paper are twofold:
we define and investigate the properties of a short rate model driven by a general Gaussian Volterra process and, after defining precisely a notion of convexity adjustment, derive explicit formulae for it.
\end{abstract}


\section{Introduction and notations}
\subsection{Introduction}
In fixed income markets, the 
different schedules of payments and the diverse currencies, margins require specific adjustments in order to price all interest rate products consistently.
This is usually referred to as convexity adjustment and has a deep impact on interest rate derivatives.
Starting from~\cite{brotherton1993yield, flesaker1993arbitrage, ritchken1993averaging},
academics and practitioners alike have developed a series of formulae for this convexity adjustment in a variety of models,
from simple stochastic rate models~\cite{kirikos1997convexity} to some incorporating stochastic volatility features~\cite{andersen2010interest}.
Recently, Garcia-Lorite and Merino~\cite{garcia2023convexity} used Malliavin calculus techniques to compute approximations of this convexity adjustment for various interest rate products.
Motivated by the new paradigm of rough volatility in Equity markets~\cite{bayer2016pricing, bonesini2023rough, el2018microstructural, fukasawa2021volatility, gatheral2022volatility, jacquier2021rough, jacquier2023deep},
we consider here stochastic dynamics for the short rate, driven by a general Gaussian Volterra process, providing more flexibility than standard Brownian motion.
In the framework of the change of measure approach in~\cite{pelsser2003mathematical}, 
we introduce a clear definition of convexity adjustment \blue{for zero coupon
bonds, in Proposition~\ref{prop:Convexity}, namely as the non-martingale correction of ratios of zero-coupon prices under the forward measure}, for which we are able to derive closed-form expressions
\blue{or asymptotic approximations}.
We introduce the model, derive its properties in Section~\ref{sec:Gaussian}.
In Section~\ref{sex:Convexity}, we define convexity adjustment and provide formulae for it, the main result of the paper, which we illustrate in some specific examples.
Section~\ref{sec:Products} provides some further expressions for liquid interest rate products, and we highlight some numerical aspects of the results in Section~\ref{sec:Numerics}.

\subsection{Model and notations}

On a given filtered probability space $(\Omega, \Ff, (\Ff_t)_{t\geq 0}, \PP)$, 
we are interested in short rate dynamics of the form
\begin{equation}\label{eq:ShortRatedWH}
r_t = \theta(t) + \int_{0}^{t}\varphi(t,u)\D \Wf_u
 = \theta(t) + \left(\varphi(t,\cdot)\circ\Wf\right)_{t},
\end{equation}
with~$\theta$ a deterministic function and~$\Wf$ a continuous Gaussian process adapted to the filtration $(\Ff_t)_{t\geq 0}$.
Here and below, given a function~$\phi$ and a stochastic process~$X$, we write
$(\phi\circ X)_{a,b} := \int_{a}^{b}\phi(s)\D X_s$,
and omit~$a$ whenever $a=0$.
For some fixed time horizon $T>0$, 
define further, for $u\leq t\leq T$,
\begin{equation}\label{eq:Phi}
\XiT(t,u) := -\int_{t}^{T}\varphi(s,u)\D s
\qquad\text{and}\qquad
\XiT(u):=\XiT(u,u)
\end{equation}
as well as $\Theta_{t,T} := \displaystyle \int_{t}^{T}\theta(s) \D s$.
We consider a given risk-neutral probability measure~$\QQ$, equivalent to~$\PP$, so that the price of the zero-coupon bond  at time~$t$ is given by
\begin{equation}\label{eq:DefZeroCoupon}
P_{t,T} := \EE^{\QQ}_{t}\left[B_{t,T}\right],
\qquad\text{where}\qquad
B_{t,T} := \exp\left\{-\int_{t}^{T}r_s \D s \right\},
\end{equation}
and we define the instantaneous forward rate process as 
\begin{equation}\label{eq:FwdRateDef}
f_{t,T}  := -\partial_{T}\log P_{t,T}.
\end{equation}

\begin{remark}
For modelling purposes, we shall consider kernels of convolution type, namely
\begin{equation}\label{eq:phiConv}
\varphi(t,u)= \varphi(t-u).
\end{equation}
\end{remark}

\subsection{Empirical motivation}
\blue{
The modelling framework above (and in particular the introduction of a potentially singular kernel)
is motivated by empirical observations.
Assume that the kernel is given by a power-law form
$\varphi(t,u) = (t-u)^{H-/\half}$ with $H \in (0,1)$, and that~$\Wf$ is a standard Brownian motion.
To estimate the Hurst exponent~$H$, we follow the methodology devised in~\cite{gatheral2022volatility}
for the instantaneous log volatility
(although more refined and robust statistical estimation techniques are now available, 
we leave a detailed empirical analysis for future work) and compute it via the linear regression
$$
\log\EE[|r_{t+\Delta} - r_{t}|^2] = 2H\log(\Delta) + c, 
\qquad\text{for }\Delta>0,
$$
for some constant~$c$.
Of course such a linear regression hinges on some assumptions on the form of~$(r_t)_{t\geq 0}$ but a detailed analysis of short rate data is beyond the scope of the present paper,
and we only provide here short insights into the potential roughness of short rates dynamics.
We consider the sport interest rate data from Option Metrics~\footnote{data available at 
\href{https://wrds-www.wharton.upenn.edu/}{WRDS/OptionMetrics}}.
We consider the data from 4/1/2010 until 28/2/2023.
For different dates within this period, Figures~\ref{fig:YC_WRDS} show the available data points (circles) as well as the interpolation by splines (the extrapolation is assumed flat).
In  Figure~\ref{fig:TimeSeriesYC_WRDS},
we compute the time series of the yield curves, 
for each (interpolated) maturities.
The estimation of the Hurst exponent for each maturity is shown in Figure~\ref{fig:Hurst_WRDS}.
\begin{figure}[H]\label{fig:YC_WRDS}
\includegraphics[scale=0.4]{"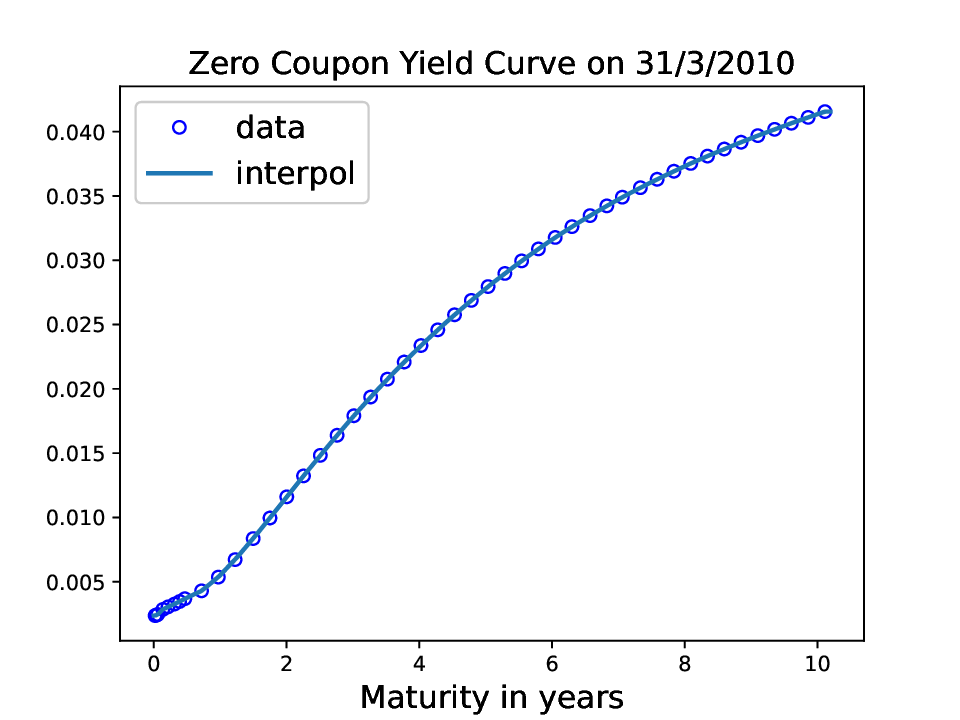"}
\includegraphics[scale=0.4]{"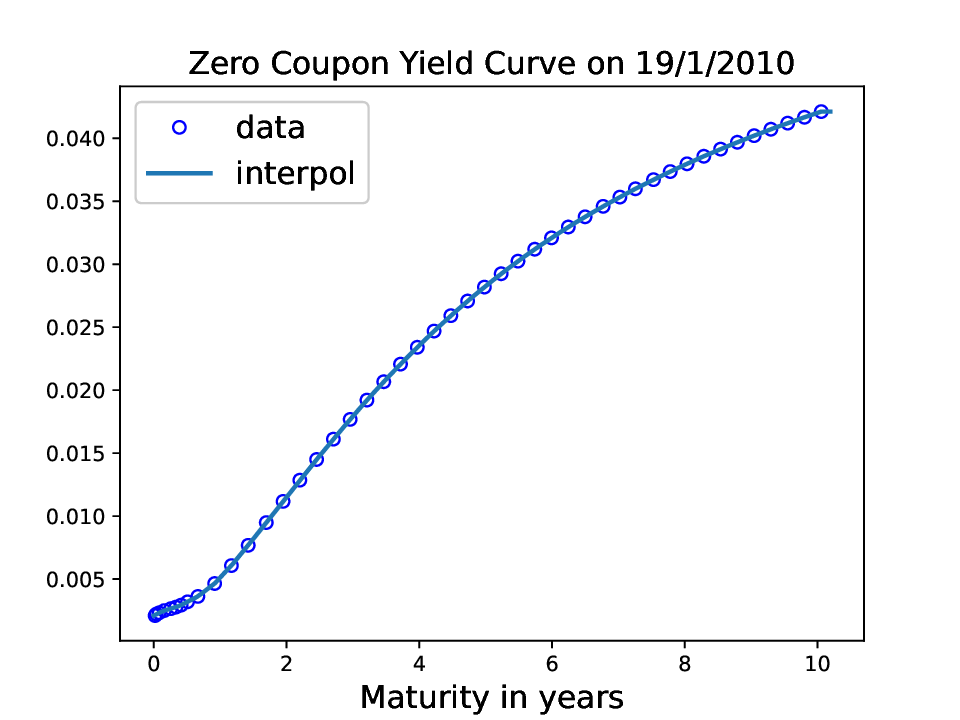"}
\caption{Examples of rates curves over different days from the OptionMetrics rates data.}
\end{figure}
\begin{figure}[H]\label{fig:TimeSeriesYC_WRDS}
\includegraphics[scale=0.5]{"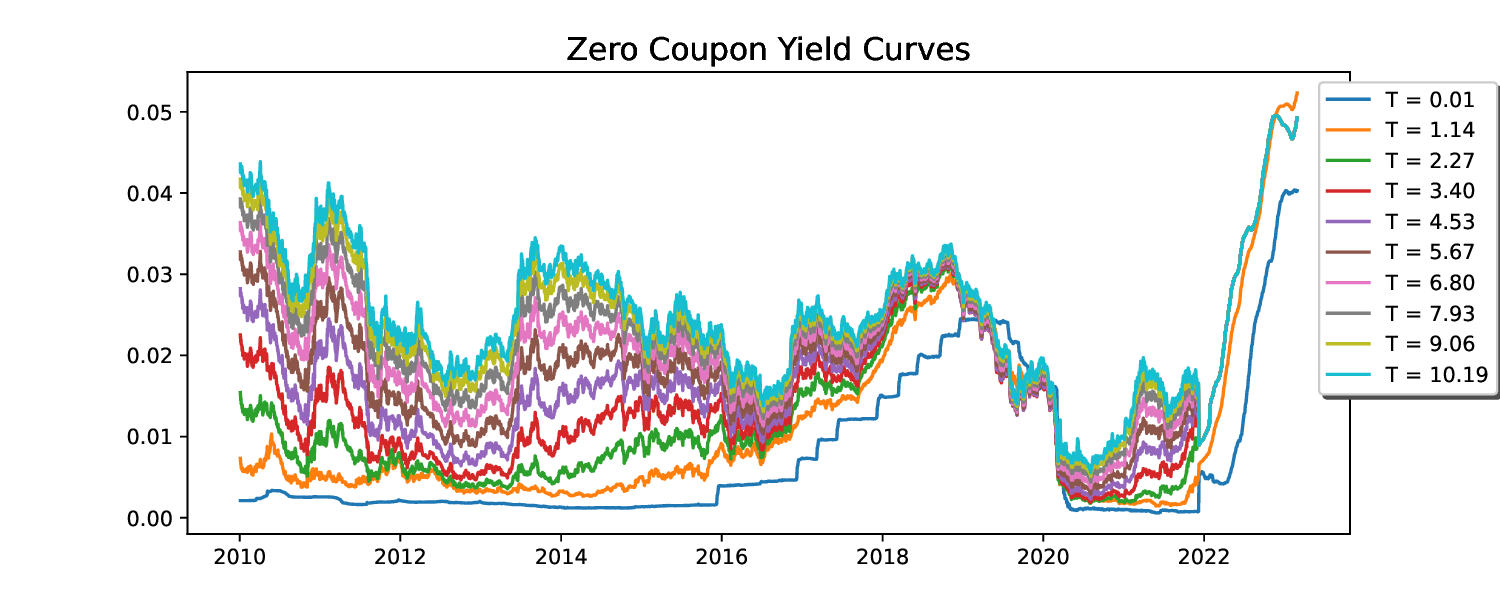"}
\caption{Time series of the OptionMetrics rates for different maturities.}
\end{figure}
\begin{figure}[H]\label{fig:Hurst_WRDS}
\includegraphics[scale=0.5]{"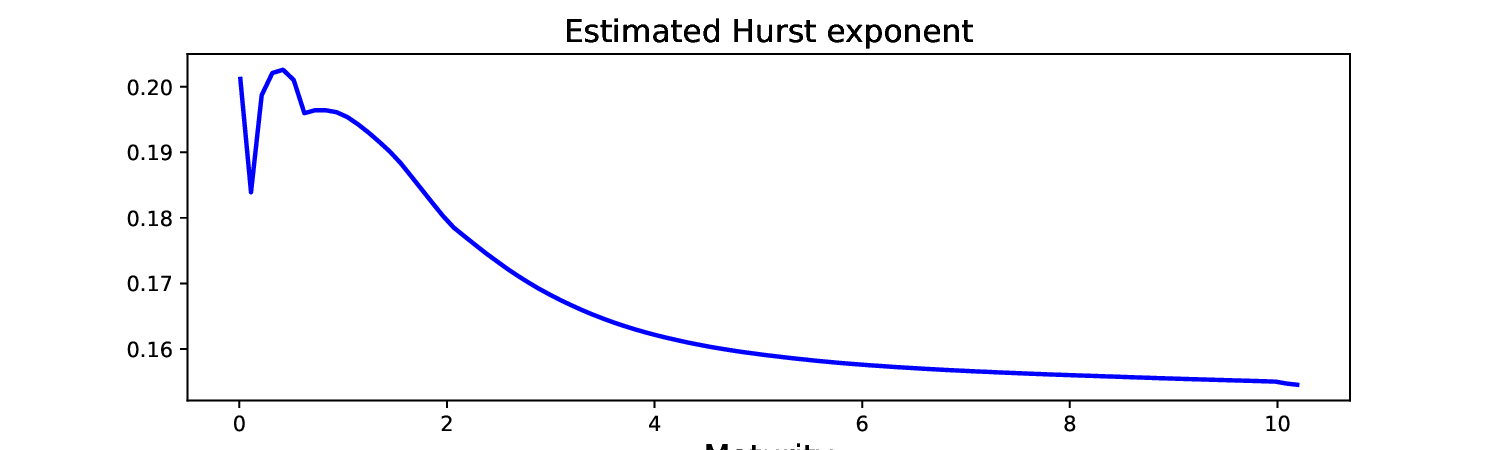"}
\caption{Estimation of the Hurst exponent for the OptionMetrics rates data.}
\end{figure}
A similar analysis on the US Daily Treasury Par Yield Curve Rates~\footnote{data available at 
\href{https://home.treasury.gov/resource-center/data-chart-center/interest-rates/TextView?type=daily_treasury_yield_curve&field_tdr_date_value=2018}{home.treasury.gov/resource-center/data-chart-center/interest-rates}} yields Figures~\ref{fig:TimeSeries_USTreasury} and~\ref{fig:Hurst_USTreasury}.
\begin{figure}[H]\label{fig:TimeSeries_USTreasury}
\includegraphics[scale=0.5]{"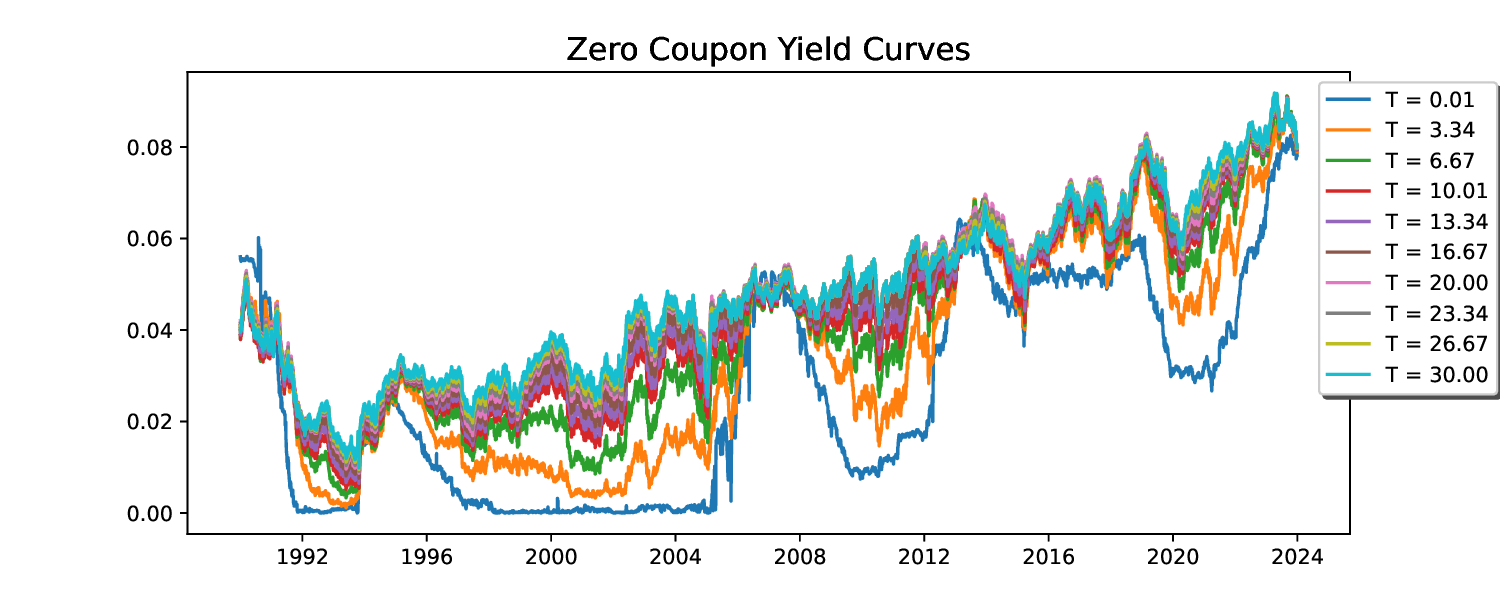"}
\caption{Time series of the  US Treasury rates for different maturities.}
\end{figure}
\begin{figure}[H]\label{fig:Hurst_USTreasury}
\includegraphics[scale=0.5]{"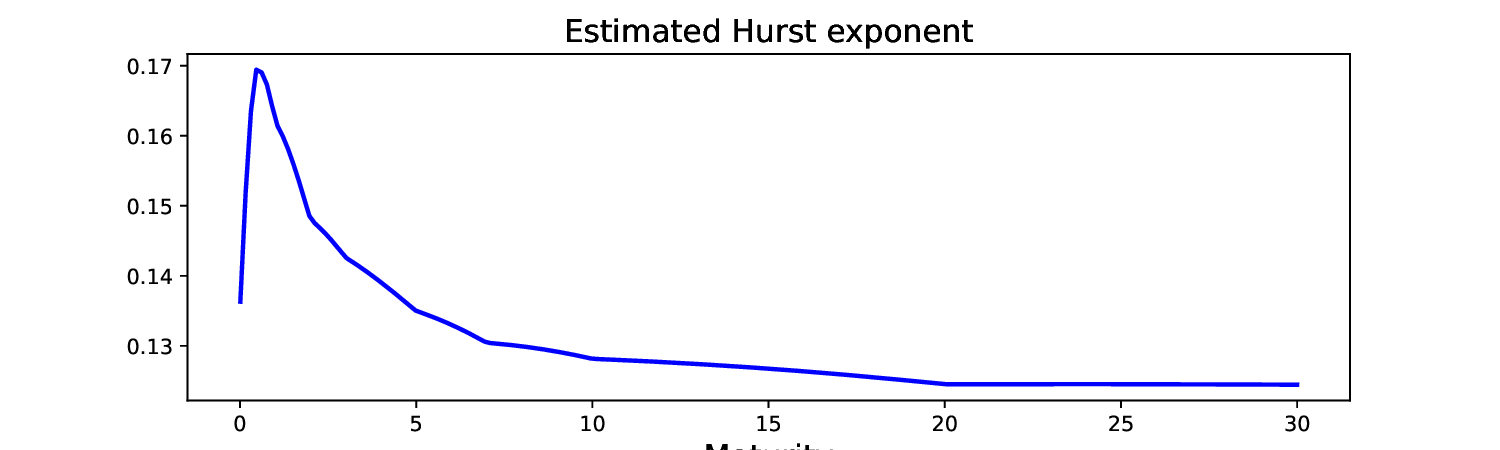"}
\caption{Estimation of the Hurst exponent for the US Treasury rates.}
\end{figure}
}
\section{Gaussian martingale driver}
\label{sec:Gaussian}

\subsection{Dynamics of the zero-coupon bond price}

We assume first that~$\Wf$ is a \blue{continuous} Gaussian martingale with $\gamma_{\Wf}(t):=\EE[\Wf_t^2]$ finite for all $t\geq 0$.
\blue{In this case, the (predictable) quadratic variation process $\gamma_{\Wf}(\cdot)$ is clearly deterministic, but also continuous and increasing, and therefore its derivative $\gamma_{\Wf}'$ exists almost everywhere.}
In order to ensure existence of the rate process in~\eqref{eq:ShortRatedWH}, we assume the following \blue{(we write $\D\lambda$ for the Lebesgue measure on~$\RR_+$)}:
\begin{assumption}\label{assu:Kernel}
For each $t \in [0,T]$, $\varphi(t,\cdot) \in L^1(\D\lambda)\cap L^2(\gamma_{\Wf})$, and~$\varphi$ is of convolution type~\eqref{eq:phiConv}.
\end{assumption}

\begin{lemma}\label{lem:GaussianSemi}
Under Assumption~\ref{assu:Kernel}, 
$\left(\XiT(t,\cdot)\circ \Wf\right)_{t}$ is an $(\Ff_{t})_{t \in [0,T]}$ Gaussian semimartingale.
\end{lemma}
\begin{proof}
From~\eqref{eq:Phi}, 
$\XiT$ is in general not in convolution form~\eqref{eq:phiConv}. 
However, since ~$\varphi$ is, we can write  
$$
\XiT(t,u) := -\int_{t}^{T}\varphi(s,u)\D s
 = -\int_{t}^{T}\varphi(s-u)\D s
 = \Phi(T-u) - \Phi(t-u),
$$
where the function~$\Phi$
is defined as
$\Phi(z) := - \int_{\cdot}^{z}\varphi(u)\D u$.
The stochastic integral then reads
$$
\left(\XiT(t,\cdot)\circ \Wf\right)_{t}
 = \int_{0}^{t}\XiT(t,u)\D \Wf_{u}
 = -\int_{0}^{t}\left[\Phi(t-u) - \Phi(T-u)\right]\D \Wf_{u},
$$
which corresponds to a two-sided moving average process in the sense of~\cite[Section 5.2]{basse2010path}.
Assumption~\ref{assu:Kernel} then implies that for each $t \in [0,T]$, the function~$\XiT(t,\cdot)$ is absolutely continuous on $[0,t]$
and $\partial_{t}\XiT(t,\cdot) \in L^2(\gamma_{\Wf})$ and the statement follows from~\cite[Theorem 5.5]{basse2010path}.
\end{proof}
\begin{remark}\ 
\begin{itemize}
\item The $L^2$ property ensures that the stochastic integral $(\varphi(t-\cdot)\circ \Wf)_t$ is well defined.
\item The assumption does not imply that the short rate itself, while Gaussian, is a semimartingale.
\end{itemize}
\end{remark}

\begin{proposition}\label{prop:ZeroCoupon}
The price of the zero-coupon bond  at time~$t$ reads 
$$
P_{t,T} 
 = 
\exp\left\{-\Theta_{t,T} + \half\int_{t}^{T}\XiT(u)^2 \D u + \left(\XiT(t,\cdot)\circ \Wf\right)_{t}\right\},
$$
and the discounted bond price
$\Ptilde_{t,T} := P_{t,T}\exp\left\{-\int_{0}^{t}r_s \D s\right\}$
is a $\QQ$-martingale satisfying
$$
\frac{\D \Ptilde_{t,T}}{\Ptilde_{t,T}} =  \XiT(t)\ \D \Wf_t.
$$
\end{proposition}

\begin{corollary}\label{cor:InstFwd}
The instantaneous forward rate satisfies $f_{TT}=r_T$ and, for all $t \in [0,T)$,
$$
f_{t,T} = \theta(T) +\int_{0}^{t}\varphi(T,u)\D \Wf_u + \int_{t}^{T}\varphi(T,u)\XiT(u)\D u.
$$
\end{corollary}
In differential form, for any fixed $T>0$, for $t \in [0,T]$, this is equivalent to
$$
\D f_{t,T} = \varphi(T-t)\D \Wf_t - \varphi(T-t)\XiT(t)\D t.
$$

\begin{algo}
For simulation purposes, we consider a time grid $\Tt:=\{0=t_0<t_1<\cdots<t_N=T\}$ and discretise the stochastic integral along this grid with left-point approximations as
$$
\left(\XiT(t_i,\cdot)\circ \Wf\right)_{t_i}
 = \int_{0}^{t_i} \XiT(t_i,u)\D \Wf_u 
\approx \sum_{k=0}^{i-1} \XiT(t_i, t_k) \left(\Wf_{t_{k+1}} - \Wf_{t_k}\right),
\qquad\text{for each }i = 1,\ldots, N.
$$
The vector 
$\left(\XiT(t_i,\cdot)\circ \Wf\right)_{t_i \in\Tt}
$ of stochastic integrals 
can then be simulated along the grid directly as
$$
\begin{pmatrix}
\left(\XiT(t_1,\cdot)\circ \Wf\right)_{t_1}\\
\vdots\\
\left(\XiT(t_N,\cdot)\circ \Wf\right)_{t_N}\\
\end{pmatrix}
 \approx
\begin{pmatrix}
\XiT(t_{1}, t_{0}) &  & & & \\
\XiT(t_{2}, t_{0}) & \XiT(t_{2}, t_{1}) & & & \\
\vdots & \ddots & \ddots & & \\
\XiT(t_{N-1}, t_{0}) & \XiT(t_{N-1}, t_{1}) & \ldots & \XiT(t_{N-1}, t_{N-2}) & \\
\XiT(t_{N}, t_{0}) & \XiT(t_{N}, t_{1}) & \ldots & \XiT(t_{N}, t_{N-1})
\end{pmatrix}
\begin{pmatrix}
\Wf_{t_{1}} - \Wf_{t_0}\\
\vdots\\
\Wf_{t_{N}} - \Wf_{t_{N-1}}
\end{pmatrix},
$$
where the middle matrix is lower triangular (we omit the null terms everywhere for clarity).
\end{algo}

\begin{example}\label{ex:Example}
With $\varphi(t) = \sigma\E^{-\kappa t}$, for $\sigma>0$, 
$\theta(t) := r_{0}\E^{-\kappa t} + \mu\left(1-\E^{-\kappa t}\right)$ and $\Wf=W$ a Brownian motion,
we recover exactly the Vasicek model~\cite{vasicek1977equilibrium}, namely
$r_t = r_0 + \kappa \int_{0}^{t}(\mu-r_s)\D s + \sigma W_t$.
\end{example}

\blue{
\begin{example}\label{ex:Example}
Consider the extension of the Vasicek model proposed by Hull and White~\cite{hull1990pricing}, where
$
\D r_{t} = (\zeta(t) - a(t)r_{t})\D t + \sigma(t)\D W_t,
$
where $\zeta(\cdot), a(\cdot)$ and~$\sigma(\cdot)$ are sufficiently smooth deterministic functions of time.
Direct computations yield the solution,
with $A(t):=\int_{0}^{t}a(s)\D s$,
$$
r_t = r_{0} + \int_{0}^{t}
\E^{-(A(t)-A(s))}\zeta(s)\D s +\int_{0}^{t}\E^{-(A(t)-A(s))}\sigma(s)\D W_s.
$$
Letting
$$
\theta(t) := r_{0} + \int_{0}^{t}
\E^{-(A(t)-A(s))}\zeta(s)\D s,
\qquad
\varphi(t,s) := \E^{-(A(t)-A(s))}
\qquad\text{and}\qquad
\D\Wf_t = \sigma(t)\D W_t
$$
makes it coincide exactly with our setup in~\eqref{eq:ShortRatedWH}.
Now Assumption~\ref{assu:Kernel} holds if and only if $A(t)-A(s) = A(t-s)$ for all $0\leq s\leq t$, 
namely when the function~$a$ is linear or constant.
Note that, as mentioned in~\cite[Section~3.3]{brigo2006interest}, the function~$a$ is often assumed constant in practice.
\end{example}
}

\begin{proof}[Proof of Proposition~\ref{prop:ZeroCoupon}]
The price of the zero-coupon bond at time~$t$ then reads
\begin{align}\label{eq:PtTH}
P_{t,T} := \EE^{\QQ}_{t}\left[\exp\left\{-\int_{t}^{T}r_s \D s \right\}\right]
&  = \EE^{\QQ}_{t}\left[\exp\left\{-\int_{t}^{T}\left(\theta(s) + \int_{0}^{s}\varphi(s,u)\D \Wf_u\right) \D s \right\}\right]\nonumber\\
 & = \E^{-\Theta_{t,T}}\EE^{\QQ}_{t}\left[\exp\left\{-\int_{t}^{T}\left(\int_{0}^{s}\varphi(s,u)\D \Wf_u \right)\D s \right\}\right].
\end{align}
Using Fubini, we can write
\begin{align}\label{eq:Fubini}
-\int_{t}^{T}\left(\int_{0}^{s}\varphi(s,u) \D \Wf_u \right)\D s 
 &  = -\int_{0}^{t}\left(\int_{t}^{T}\varphi(s,u)\D s \right)\D \Wf_u
 - \int_{t}^{T}\left(\int_{u}^{T}\varphi(s,u)\D s \right)\D \Wf_u\\
 &  = \int_{0}^{t}\XiT(t,u)\D \Wf_u
 + \int_{t}^{T}\XiT(u)\D \Wf_u,\nonumber
\end{align}
using~\eqref{eq:Phi}.
Plugging this into~\eqref{eq:PtTH}, the zero-coupon bond then reads
\begin{align*}
P_{t,T} 
 &  = \E^{-\Theta_{t,T}} \exp\left\{\int_{0}^{t}\XiT(t,u)\D \Wf_u\right\}
\EE^{\QQ}_{t}\left[ \exp\left\{\int_{t}^{T}\XiT(u)\D \Wf_u\right\}\right]\\
 & = \E^{-\Theta_{t,T}} \exp\Big\{\left(\XiT(t,\cdot)\circ \Wf\right)_{t}\Big\}
\EE^{\QQ}_{t}\left[ \E^{\left(\XiT\circ \Wf\right)_{t,T}}\right].
\end{align*}
Conditional on~$\Ff_t$,  $(\XiT\circ \Wf)_{t,T}$ is centered Gaussian with
$\VV_t[(\XiT\circ \Wf)_{t,T}] = \int_{t}^{T}\XiT(u)^2 \D u$,
hence
$
P_{t,T} 
 = \E^{-\Theta_{t,T}}
\exp\left\{\left(\XiT(t,\cdot)\circ \Wf\right)_{t} + \half\int_{t}^{T}\XiT(u)^2 \D u\right\}.
$
By Fubini and Assumption~\ref{assu:Kernel},
\begin{align*}
\left(\XiT(t,\cdot)\circ \Wf\right)_{t} 
 = \int_{0}^{t}\XiT(t,u)\D \Wf_u 
 & = \int_{0}^{t}\left(\XiT(u) + \int_{u}^{t}\partial_{s}\XiT(s,u)\D s\right)\D \Wf_u\\
 & = \int_{0}^{t}\XiT(u)\D \Wf_u + \int_{0}^{t}\int_{u}^{t}\partial_{s}\XiT(s,u)\D s\D \Wf_u\\
 & = \int_{0}^{t}\XiT(u)\D \Wf_u + \int_{0}^{t}\int_{0}^{s}\partial_{s}\XiT(s,u)\D \Wf_u\D s\\
 & = \int_{0}^{t}\XiT(u)\D \Wf_u + \int_{0}^{t}\int_{0}^{s}\varphi(s,u)\D \Wf_u\D s.
\end{align*}
This is an $L^1$-Dirichlet process~\cite[Definition 2]{russo2006bifractional},
written as a decomposition of a local martingale and a term with zero quadratic variation. 
Therefore $\langle \log(P_{\cdot,T}), \log(P_{\cdot,T})\rangle_t = \int_{0}^{t}\XiT(u)^2\D u$ and
\begin{equation}\label{eq:dXHalf}
\D \log(P_{\cdot,T})
 = \left(\theta(t) + \left(\partial_t\XiT(t,\cdot)\circ \Wf\right)_{t} - \half\XiT(t)^2\right)\D t + \XiT(t)\D \Wf_t.
\end{equation}
Now, It\^o's formula \blue{with $X_t := \log(P_{t,T})$}, using~\eqref{eq:dXHalf} yields
$P_{T,T} = P_{t,T} + \int_{t}^{T}P_{s,T}\D X_s + \half\int_{t}^{T}P_{s,T}\D\langle X,X\rangle_{s}$, hence, for each $T>0$,
$\D P_{T,T} = \D P_{t,T} -P_{t,T}\D X_t - \half P_{t,T}\D\langle X,X\rangle_{t}$,
and therefore, since $P_{T,T}=1$,
\begin{align*}
\frac{\D P_{t,T}}{P_{t,T}} & = \D X_t + \half \D\langle X,X\rangle_{t}
   = \left(\underbrace{\theta(t) + \left(\partial_t\XiT(t,\cdot)\circ \Wf\right)_{t}}_{r_t} - \half\XiT(t)^2\right)\D t + \XiT(t)\D \Wf_t
+\half \D\left(\int_{0}^{t}\XiT(u)^2 \D u\right)\\
 &  = r_t\D t + \XiT(t)\D \Wf_t - \half\XiT(t)^2\D t
+\half\XiT(t)^2\D t
  = r_t\D t + \XiT(t)\D \Wf_t.
\end{align*}
The dynamics of the discounted zero-coupon bond price in the lemma follows immediately.
\end{proof}

\begin{proof}[Proof of Corollary~\ref{cor:InstFwd}]
It follows by direct computation starting from the instantaneous forward rate~\eqref{eq:FwdRateDef}:
\begin{align*}
f_{t,T} 
 & = \partial_{T}\Theta_{t,T} -\partial_{T} \int_{0}^{t}\XiT(t,u)\D \Wf_{u} - \half\partial_{T}\int_{t}^{T}\XiT(u)^2 \D u\\
 & = \partial_{T}\Theta_{t,T} -\partial_{T} \int_{0}^{t}\left(-\int_{t}^{T}\varphi(s,u)\D s\right)\D \Wf_u - \half\partial_{T}\int_{t}^{T}\left(-\int_{u}^{T}\varphi(s,u)\D s\right)^2 \D u\\
 & = \theta(T) +\int_{0}^{t}\partial_{T}\left(\int_{t}^{T}\varphi(s,u)\D s\right)\D \Wf_u - \half\partial_{T}\int_{t}^{T}\left(\int_{u}^{T}\varphi(s,u)\D s\right)^2 \D u\\
 & = \theta(T) +\int_{0}^{t}\varphi(T,u)\D \Wf_u - \half\left(\int_{T}^{T}\varphi(s,T)^2\D s + \int_{t}^{T}\partial_{T}\left[\left(\int_{u}^{T}\varphi(s,u)\D s\right)^2\right] \D u\right)
\\
 & = \theta(T) +\int_{0}^{t}\varphi(T,u)\D \Wf_u - \int_{t}^{T}\varphi(T,u)\left(\int_{u}^{T}\varphi(s,u)\D s\right) \D u\\
 & = \theta(T) +\int_{0}^{t}\varphi(T,u)\D \Wf_u + \int_{t}^{T}\varphi(T,u)\XiT(u)\D u.
\end{align*}
\end{proof}

\begin{remark}
The two lemmas above correspond to the two sides of the Heath-Jarrow-Morton framework.
From the expression of the instantaneous forward rate, let
$\alpha_{t,T} := \varphi(T-t)\XiT(t)$
and $\beta_{t,T} := \varphi(T-t)$,
so that
$\D f_{t,T} = \beta_{t,T}\D \Wf_t -\alpha_{t,T}\D t$,
and consider the discounted bond price
$$
\Ptilde_{t,T} := P_{t,T}\exp\left\{-\int_{0}^{t}r_s \D s\right\}
 = \exp\left\{-\int_{0}^{t}r_s \D s - \int_{t}^{T}f_{t,s}\D s\right\}
 =: \E^{Z_t}.
$$
It\^os' formula then yields
\begin{equation}\label{eq:SDEPTilde}
\frac{\D \Ptilde_{t,T}}{\Ptilde_{t,T}} = \D Z_t + \half \D\langle Z,Z\rangle_t.
\end{equation}
From the differential form of $f_{t,T}$, we can write, for any $t \in [0,T)$,
$$
f_{t,T} = f_{0,T} + \int_{0}^{t}\D f_{s,T}
 = f_{0,T} + \int_{0}^{t}\Big(\varphi(T,u)\D \Wf_u - \varphi(T,u)\XiT(u)\D u\Big)
 = f_{0,T} + \int_{0}^{t}\beta_{u,T}\D \Wf_u + \int_{0}^{t}\alpha_{u,T}\D u,
$$
so that, using stochastic Fubini, we obtain
\begin{align*}
F_{t,T} & := \int_{t}^{T}f_{t,s}\D s
 =  \int_{t}^{T}\left(f_{0,s} + \int_{0}^{t}\beta_{u,s}\D \Wf_u + \int_{0}^{t}\alpha_{u,s}\D u\right)\D s\\
&  =  \int_{t}^{T}f_{0,s}\D s + \int_{0}^{t}\int_{t}^{T}\beta_{u,s}\D s\D \Wf_u + \int_{0}^{t}\int_{t}^{T}\alpha_{u,s}\D s\D u.
\end{align*}
Now, 
\begin{align*}
\int_{t}^{T}f_{0,s}\D s
 & = \int_{t}^{T}\left(f_{s,s} - \int_{0}^{s}\partial_{u} f_{u,s} \D u\right)\D s\\
 & = \int_{t}^{T}r_s \D s - \int_{0}^{t} \int_{t}^{T}\partial_{u} f_{u,s} \D s \D u - \int_{t}^{T} \int_{u}^{T}\partial_{u} f_{u,s} \D s \D u\\
 & = \int_{t}^{T}r_s \D s - \int_{0}^{t} \left(\int_{t}^{T}\partial_{u} f_{u,s} \D s - \int_{u}^{T}\partial_{u} f_{u,s} \D s \right)\D u - \int_{0}^{T} \int_{u}^{T}\partial_{u} f_{u,s} \D s \D u\\
 & = \int_{t}^{T}r_s \D s + \int_{0}^{t} \int_{u}^{t}\partial_{u} f_{u,s} \D s \D u - \int_{0}^{T} \int_{u}^{T}\partial_{u} f_{u,s} \D s \D u,
\end{align*}
using Fubini, so that
$$
F_{t,T} 
 =  \underbrace{\int_{t}^{T}r_s \D s + \int_{0}^{t} \int_{u}^{t}\partial_{u} f_{u,s} \D s \D u - \int_{0}^{T} \int_{u}^{T}\partial_{u} f_{u,s} \D s \D u}_{\int_{t}^{T}f_{0,s}\D s}
 + \int_{0}^{t}\int_{t}^{T}\beta_{u,s}\D s\D \Wf_u + \int_{0}^{t}\int_{t}^{T}\alpha_{u,s}\D s\D u,
$$
and
$\D F_{t,T} 
 =  \left(\int_{t}^{T}\alpha_{t,s}\D s - r_t\right)\D t + \left(\int_{t}^{T}\beta_{t,s}\D s\right)\D \Wf_t$.
Therefore,
$$
\D Z_t = \D\left(-\int_{0}^{t}r_s \D s - \int_{t}^{T}f_{t,s}\D s\right)
 = -r_t \D t - \D F_{t,T}
 = -r_t \D t - \D F_{t,T}
=  -\left(\int_{t}^{T}\alpha_{t,s}\D s \right)\D t - \left(\int_{t}^{T}\beta_{t,s}\D s\right)\D \Wf_t,
$$
and~\eqref{eq:SDEPTilde} gives
$$
\frac{\D \Ptilde_{t,T}}{\Ptilde_{t,T}} = -\left(\int_{t}^{T}\alpha_{t,s}\D s - \half \left(\int_{t}^{T}\beta_{t,s}\D s\right)^2\right)\D t - \left(\int_{t}^{T}\beta_{t,s}\D s\right)\D \Wf_t.
$$
The discounted process $(\Ptilde_{t,T})_{t \in [0,T]}$ is  a local martingale if and only if its drift is null:
for $t\in (0,T)$,
$$
\partial_{T}\left\{\int_{t}^{T}\alpha_{t,s}\D s - \half \left[\int_{t}^{T}\beta_{t,s}\D s\right]^2\right\}
  = \alpha_{t,T} - \beta_{t,T}\int_{t}^{T}\beta_{t,s}\D s
 = \varphi(T-t)\left[\XiT(t) - \int_{t}^{T}\varphi(s,t)\D s\right],
$$
which is equal to zero by definition of the functions.
Therefore the drift (as a function of~$T$) is constant. 
Since it is trivially equal to zero at $T=t$, it is null everywhere 
and $(\Ptilde_{t,T})_{t \in [0,T]}$ is 
a $\QQ$-local martingale.
\end{remark}

\subsection{Convexity adjustments}
\label{sex:Convexity}
We now enter the core of the paper, investigating the influence of the Gaussian driver on the convexity of bond prices.
We first start with the following simple proposition:
\begin{proposition}\label{cor:SDEProd}
For any $T, \tau\geq 0$,
\begin{align*}
\D \left(\frac{1}{P_{t,\tau}}\right)
 &  = \frac{\left(\Xi_{\tau}(t,t)^2\gamma'_{\Wf}(t) - r_t\right) \D t}{P_{t,\tau}} - \frac{\Xi_{\tau}(t,t)}{P_{t,\tau}}\D \Wf_t,\\
\D\left(\frac{P_{t,T}}{P_{t,\tau}}\right)
 & = \frac{P_{t,T}}{P_{t,\tau}}\Big(\Xi_{T}(t,t)- \Xi_{\tau}(t,t)\Big)\Big\{-\Xi_{\tau}(t,t)\gamma'_{\Wf}(t)\D t  + \D \Wf_t \Big\},
\end{align*}
and there exists a probability measure~$\QQ^{\tau}$ such that
$\Wf^{\QQ^{\tau}}_t$ is a $\QQ^{\tau}$-Gaussian martingale and
\begin{equation}\label{eq:DynamicsRatioQ}
\D\left(\frac{P_{t,T}}{P_{t,\tau}}\right) = 
\frac{P_{t,T}}{P_{t,\tau}}\Sigma_{t}^{T,\tau} \D \Wf^{\QQ^{\tau}}_t,
\end{equation}
under $\QQ^{\tau}$, 
where $\Sigma_{t}^{T,\tau}
:= \Xi_{T}(t) - \Xi_{\tau}(t)
= \Phi(T-t) - \Phi(\tau-t)$.
\end{proposition}
Note that, from the definition of $\Xi_{T}$
in~\eqref{eq:Phi}, 
$\Sigma_{t}^{T,\tau}$ is non-negative whenever $\tau\geq T$.
\blue{In standard Fixed Income literature, the probability measure~$\QQ^{\tau}$ corresponds to the $\tau$-forward measure.}

\begin{proof}
From the definition of the zero-coupon price~\eqref{eq:DefZeroCoupon} and Proposition~\ref{prop:ZeroCoupon}, 
$P_{t,T}$ is strictly positive almost surely and 
$$
\frac{\D P_{t,T}}{P_{t,T}} = r_t \D t + \Xi_{T}(t,t)\D \Wf_t, 
$$
and therefore It\^o's formula implies that, for any $0\leq t \leq \tau$,
$$
\D\left(\frac{1}{P_{t,\tau}}\right)
 = -\frac{\D P_{t,\tau}}{P_{t,\tau}^2} + \frac{\D\langle P_{t,\tau},P_{t,\tau}\rangle}{P_{t,\tau}^3}
 = \frac{\Big(\Xi_{\tau}(t,t)^2\gamma'_{\Wf}(t) - r_t\Big) \D t}{P_{t,\tau}} - \frac{\Xi_{\tau}(t,t)\D \Wf_t}{P_{t,\tau}}.
$$
Therefore
\begin{align*}
\D\left(\frac{P_{t,T}}{P_{t,\tau}}\right)
 & = P_{t,T}\D\left(\frac{1}{P_{t,\tau}}\right) + \frac{\D P_{t,T}}{P_{t,\tau}} + \D P_{t,T}\cdot\D\left(\frac{1}{P_{t,\tau}}\right)\\
 & = \frac{P_{t,T}}{P_{t,\tau}}\left\{\Big(\Xi_{\tau}(t,t)^2\gamma'_{\Wf}(t) - r_t\Big) \D t - \Xi_{\tau}(t,t)\D \Wf_t
 + \Big( r_t \D t + \Xi_{T}(t,t)\D \Wf_t\Big) - \Xi_{T}(t,t)\Xi_{\tau}(t,t)\gamma_{\Wf}'(t)\D t\right\}\\
 & = \frac{P_{t,T}}{P_{t,\tau}}\Big(\Xi_{T}(t,t)- \Xi_{\tau}(t,t)\Big)\Big\{-\Xi_{\tau}(t,t)\gamma'_{\Wf}(t)\D t  + \D \Wf_t \Big\}.
\end{align*}
Define now the Dol\'eans-Dade exponential
$$
M_t := \exp\left\{\int_{0}^{t}\Xi_{\tau}(s,s)\gamma'_{\Wf}(s)\D \Wf_s - \frac{1}{2}\int_{0}^{t}\left[\Xi_{\tau}(s,s)\gamma'_{\Wf}(s)\right]^2\D s\right\},
$$
and the Radon-Nikodym derivative
$\frac{\D\QQ^{\tau}}{\D\PP} := M$.
Girsanov's Theorem~\cite[Theorem 8.6.4]{oksendal2003stochastic} implies that 
$\Wf^{\QQ^{\tau}}_t := \Wf_t - \int_{0}^{t}\Xi_{\tau}(s,s)\gamma'_{\Wf}(s)\D s$
is a Gaussian martingale and 
$\frac{P_{t,T}}{P_{t,\tau}}$ satisfies~\eqref{eq:DynamicsRatioQ} under~$\QQ^{\tau}$.
\end{proof}

The following proposition is key and provides a closed-form expression for the convexity adjustments:
\begin{proposition}\label{prop:Convexity}
For any $\tau\geq 0$ let $\tf_1, \tf_2\geq 0$. 
We then have
$$
\EE^{\QQ^{\tau}}\left[\frac{P_{t,\tf_{1}}}{P_{t,\tf_{2}}}\right]
 = \frac{P_{0,\tf_{1}}}{P_{0,\tf_{2}}}\Cf_t^{\tau}(\tf_1, \tf_2),
 \qquad\text{for any }t \in [0, \tf_1\wedge \tf_2],
$$
where 
$\displaystyle 
\Cf_t^{\tau}(\tf_1, \tf_2) := \exp\left\{
 \int_{0}^{t}\left(\Sigma_s^{\tf_2,\tau} - \Sigma_s^{\tf_1,\tau}\right)
\Sigma_s^{\tf_2,\tau}\gamma'_{\Wf}(s)\D s\right\}$
is the convexity adjustment factor.
\end{proposition}
\begin{remark}\ 
\begin{itemize}
\item When $t=0$ or $\tf_1=\tf_2$ or $\displaystyle \frac{P_{t,\tf_{1}}}{P_{t,\tf_{2}}}$ is constant, 
there is no convexity adjustment,
i.e. $\Cf_t^{\tau}(\tf_1, \tf_2)=1$.
\item More interestingly, if $\tf_2=\tau$, then 
$\Sigma_{t}^{\tf_2,\tau} = \Sigma_{t}^{\tf_2,\tf_2} = \Xi_{\tf_2}(t,t)- \Xi_{\tf_2}(t,t) = 0$ and
$$
\Cf_t^{\tau}(\tf_1, \tf_2) = 
\Cf_t^{\tf_2}(\tf_1, \tf_2)
= \exp\left\{
 \int_{0}^{t}\left(\Sigma_s^{\tf_2,\tf_2} - \Sigma_s^{\tf_1,\tf_2}\right)
\Sigma_s^{\tf_2,\tf_2}\gamma'_{\Wf}(s)\D s\right\} = 1,
$$
and the process
$\displaystyle\left(\frac{P_{t,\tf_{1}}}{P_{t,\tf_{2}}}\right)_{t\geq 0}$ is a $\QQ^{\tau}$-martingale on $[0, \tf_1\wedge \tf_2]$.
\item Regarding the sign of the convexity adjustment, we have
\begin{align*}
\Sigma_s^{\tf_2,\tau} - \Sigma_s^{\tf_1,\tau}
 & = \Big(\Xi_{\tf_2}(s,s) - \Xi_{\tau}(s,s)\Big) - \Big(\Xi_{\tf_1}(s,s)- \Xi_{\tau}(s,s)\Big)\\
  & = \Xi_{\tf_2}(s,s) - \Xi_{\tf_1}(s,s)\\
  & = -\int_{s}^{\tf_2}\varphi(z,s)\D z + 
  \int_{s}^{\tf_2}\varphi(z,s)\D z
   = -\int_{\tf_1}^{\tf_2}\varphi(z,s)\D z.
\end{align*}
Since $\varphi(\cdot)$ is strictly positive, then 
$\sgn(\Sigma_s^{\tf_2,\tau} - \Sigma_s^{\tf_1,\tau}) = \sgn(\tf_1 - \tf_2)$.
Furthermore, since 
$$
\Sigma_s^{\tf_2,\tau}
 = \Xi_{\tf_2}(s,s)- \Xi_{\tau}(s,s)
= -\int_{s}^{\tf_2}\varphi(z,s)\D z
+ \int_{s}^{\tau}\varphi(z,s)\D z
 = \int_{\tf_2}^{\tau}\varphi(z,s)\D z,
$$
then $\sgn(\Sigma_s^{\tf_2,\tau}) = \sgn(\tau-\tf_2)$, 
and therefore, assuming $\gamma'_{\Wf}$ strictly positive (as will be the case in all the examples considered here),
\begin{center}
\begin{tabular}{ |c|c|c| } 
 \hline
 $\sgn(\log\Cf_t^{\tau}(\tf_1, \tf_2))$ &  $\tf_1 > \tf_2$ & $\tf_1 < \tf_2$ \\ 
 \hline
 $\tau<\tf_2$ & negative & positive \\ 
 $\tau>\tf_2$ & positive & negative \\ 
 \hline
\end{tabular}
\end{center}
Considering without generality $\tf_1<\tf_2$, the convexity adjustment
is therefore greater than~$1$ for $\tau<\tf_2$
and less than~$1$ above.
\end{itemize}
\end{remark}

\begin{proof}[Proof of Proposition~\ref{prop:Convexity}]
Under~$\QQ^{\tau}$, the process defined as $X_t := P_{t,T} / P_{t,\tau}$
satisfies $\D X_t = X_t \Sigma_t^{T,\tau} \D\Wf_t^{\QQ^{\tau}}$, is clearly lognormal 
 and hence It\^o's formula implies
$$
\D\log(X_t)   = \frac{\D X_t}{X_t} - \frac{1}{2}\frac{\D\langle X,X\rangle_t}{X_t^2}
  = \Sigma_t^{T,\tau}\D \Wf_t^{\QQ^{\tau}} - \frac{1}{2}\left(\Sigma_t^{T,\tau}\right)^2 \gamma'_{\Wf}(t) \D t,
$$
so that
$$
X_t = X_0 \exp\left\{\int_{0}^{t}\Sigma_s^{T,\tau}\D \Wf_s - \frac{1}{2}\int_{0}^{t}\left(\Sigma_s^{T,\tau}\right)^2 \gamma'_{\Wf}(s) \D s\right\},
$$
and therefore
$$
\frac{P_{t,T}}{P_{t,\tau}} 
= \frac{P_{0,T}}{P_{0,\tau}} \exp\left\{\int_{0}^{t}\Sigma_s^{T,\tau}\D \Wf_s - \frac{1}{2}\int_{0}^{t}\left(\Sigma_s^{T,\tau}\right)^2 \gamma'_{\Wf}(s) \D s\right\}.
$$
With successively $T=\tf_1$ and $T=\tf_2$, we can then write
\begin{align*}
\frac{P_{t,\tf_1}}{P_{t,\tau}} 
 & = \frac{P_{0,\tf_1}}{P_{0,\tau}} \exp\left\{\int_{0}^{t}\Sigma_s^{\tf_1,\tau}\D \Wf_s - \frac{1}{2}\int_{0}^{t}\left(\Sigma_s^{\tf_1,\tau}\right)^2 \gamma'_{\Wf}(s) \D s\right\},\\
\frac{P_{t,\tf_2}}{P_{t,\tau}} 
 & = \frac{P_{0,\tf_2}}{P_{0,\tau}} \exp\left\{\int_{0}^{t}\Sigma_s^{\tf_2,\tau}\D \Wf_s - \frac{1}{2}\int_{0}^{t}\left(\Sigma_s^{\tf_2,\tau}\right)^2 \gamma'_{\Wf}(s) \D s\right\},
 \end{align*}
so that
\begin{align*}
\frac{P_{t,\tf_1}}{P_{t,\tf_2}} 
 & = \frac{P_{0,\tf_1}}{P_{0,\tf_2}} \exp\left\{\int_{0}^{t}\Sigma_s^{\tf_1,\tau}\D \Wf_s - \frac{1}{2}\int_{0}^{t}\left(\Sigma_s^{\tf_1,\tau}\right)^2 \gamma'_{\Wf}(s) \D s
 - \int_{0}^{t}\Sigma_s^{\tf_2,\tau}\D \Wf_s + \frac{1}{2}\int_{0}^{t}\left(\Sigma_s^{\tf_2,\tau}\right)^2 \gamma'_{\Wf}(s) \D s
 \right\}\\
 & = \frac{P_{0,\tf_1}}{P_{0,\tf_2}} \exp\left\{
 \int_{0}^{t}\left(\Sigma_s^{\tf_1,\tau} - \Sigma_s^{\tf_2,\tau}\right)\D \Wf_s 
 + \frac{1}{2}\int_{0}^{t}\left[\left(\Sigma_s^{\tf_2,\tau}\right)^2 - \left(\Sigma_s^{\tf_1,\tau}\right)^2\right]\gamma'_{\Wf}(s) \D s \right\}\\
 & = \frac{P_{0,\tf_1}}{P_{0,\tf_2}} \exp\left\{
 \int_{0}^{t}\left(\Sigma_s^{\tf_1,\tau} - \Sigma_s^{\tf_2,\tau}\right)\D \Wf_s 
 - \half \int_{0}^{t}\left(\Sigma_s^{\tf_1,\tau} - \Sigma_s^{\tf_2,\tau}\right)^2\gamma'_{\Wf}(s)\D s\right\}\\
 & \qquad \exp\left\{
 \half \int_{0}^{t}\left[\left(\Sigma_s^{\tf_1,\tau}\right)^2
+ \left(\Sigma_s^{\tf_2,\tau}\right)^2
-  2\Sigma_s^{\tf_1,\tau}\Sigma_s^{\tf_2,\tau}\right]\gamma'_{\Wf}(s)\D s 
 + \half \int_{0}^{t}\left[\left(\Sigma_s^{\tf_2,\tau}\right)^2 - \left(\Sigma_s^{\tf_1,\tau}\right)^2\right]\gamma'_{\Wf}(s) \D s \right\}\\
 & = \frac{P_{0,\tf_1}}{P_{0,\tf_2}} \exp\left\{
 \int_{0}^{t}\left(\Sigma_s^{\tf_1,\tau} - \Sigma_s^{\tf_2,\tau}\right)\D \Wf_s 
 - \half \int_{0}^{t}\left(\Sigma_s^{\tf_1,\tau} - \Sigma_s^{\tf_2,\tau}\right)^2\gamma'_{\Wf}(s)\D s\right\}\\
 & \qquad \exp\left\{
 \int_{0}^{t}\left[\left(\Sigma_s^{\tf_2,\tau}\right)^2
-  \Sigma_s^{\tf_1,\tau}\Sigma_s^{\tf_2,\tau}\right]\gamma'_{\Wf}(s)\D s\right\}. 
\end{align*}
The first exponential is a Dol\'eans-Dade exponential martingale under~$\QQ^{\tau}$, thus has 
$\QQ^{\tau}$-expectation equal to one, and the proposition follows.
\end{proof}

\subsection{Examples}
Let $\Wf = W$ be a standard Brownian motion,
so that $\gamma_{\Wf}(t)=t$ and $\gamma'_{\Wf}(t)=1$.

\subsubsection{Exponential kernels}\label{sec:Ex_ExpKernel}
Assume that $\varphi(t) = \E^{-\alpha t}$ for some $\alpha>0$, then the short rate process is of Ornstein-Uhlenbeck type and
$$
\XiT(t,u) = \Phi(T-u)  - \Phi(t-u) 
\qquad\text{with}\qquad
\Phi(z) := \frac{1}{\alpha}\E^{-\alpha z}.
$$
We can further compute 
$\Xi_{\tau}(t,t) = \Phi(\tau,t)  - \Phi(t,t)$, and
$$
\Sigma_{t}^{T,\tau}
 = \Xi_{T}(t,t) - \Xi_{\tau}(t,t)
 = \Phi(T,t)  - \Phi(t,t) - \Phi(\tau,t)  + \Phi(t,t)
 = \Phi(T,t)  - \Phi(\tau,t).
$$
Therefore the diffusion coefficient~$\Sigma_{t}^{T,\tau}$ 
and the Girsanov drift~$\Xi_{\tau}(t,t)$ read
$$
\Xi_{\tau}(t,t) = \frac{1}{\alpha}\left(\E^{-\alpha(\tau-t)} - 1\right)
\qquad\text{and}\qquad
\Sigma_{t}^{T,\tau} = \frac{1}{\alpha}\left(\E^{-\alpha(T-t)} - \E^{-\alpha(\tau-t)}\right).
$$
Finally, regarding the convexity adjustment,
$$
\log \Cf_t^{\tau}(\tf_1, \tf_2) = 
\frac{\E^{2\alpha t} - 1}{2\alpha^3}\left\{
\left(\E^{-\alpha \tf_1} - \E^{-\alpha \tf_2}\right)\E^{-\alpha \tau}
 + \E^{-2\alpha \tf_2} - \E^{-\alpha (\tf_1+\tf_2)}
 \right\}.
$$
Note that, as~$\alpha$ tends to zero, 
namely $r_t = \theta(t) +  W_t$ (in the limit), we obtain
$$
\Cf_t^{\tau}(\tf_1, \tf_2) = 
\exp\Big\{(\tf_2-\tf_1)(\tf_2-\tau)t\Big\}.
$$

\subsubsection{Riemann-Liouville kernels}\label{sec:Ex_RLKernel}
Let $H \in (0, 1)$ and $H_{\pm} := H\pm\half$.
If $\varphi(t) = t^{\Hm}$, with , the short rate process~\eqref{eq:ShortRatedWH} is driven by a Riemann-Liouville fractional Brownian motion with Hurst exponent~$H$.
Furthermore, with $\Hp := H+\half$,
$$
\XiT(t,u) = \Phi(T-u)  - \Phi(t-u) 
\qquad\text{with}\qquad
\Phi(z) := -\frac{z^{\Hp}}{\Hp}.
$$
Therefore the diffusion coefficient~$\Sigma_{t}^{T,\tau}$ and Girsanov drift~$\Xi_{\tau}(t,t)$ read
$$
\Xi_{\tau}(t,t) = -\frac{(\tau-t)^{\Hp}}{\Hp}
\qquad\text{and}\qquad
\Sigma_{t}^{T,\tau} = \frac{(\tau-t)^{\Hp} - (T-t)^{\Hp}}{\Hp}.
$$
Regarding the convexity adjustment, we instead have
\begin{align*}
\Cf_t^{\tau}(\tf_1, \tf_2) 
 & = \exp\left\{
\int_{0}^{t}\left(\Sigma_s^{\tf_2,\tau} - \Sigma_s^{\tf_1,\tau}\right)
\Sigma_s^{\tf_2,\tau}\D s\right\}
\end{align*}
Unfortunately, there does not seem to be a closed-form simplification here. We can however provide the following approximations:
\begin{lemma}
The following asymptotic expansions are straightforward and provide some closed-form expressions that may help the reader grasp a flavour on the roles of the parameters:
\begin{itemize}
\item As $t$ tends to zero,
$$
\log \Cf_t^{\tau}(\tf_1, \tf_2) 
= \frac{t}{\Hp^2}
\left(\tf_2^{\Hp} - \tf_1^{\Hp}\right)
\left(\tf_2^{\Hp} - \tau^{\Hp}\right)
 + \Oo\left(t^2\right).
$$
\item For any $\eta>0$, as $\eps$ tend to zero,
$$
\log \Cf_t^{\tf_1-\eps}(\tf_1, \tf_1+\eps) 
 = \frac{1+\eta}{2H}
\Big(\tf_1^{2H} - (\tf_1-t)^{2H}\Big)\eps^2
 + \Oo\left(\eps^3\right).
 $$
\end{itemize}
\end{lemma}
\begin{proof}
From the explicit computation of 
$\Sigma_{t}^{T,\tau}$ above, we can write, as~$s$ tends to zero,
$$
\Sigma_s^{T,\tau}
 = \frac{(\tau-s)^{\Hp} - (T-s)^{\Hp}}{\Hp}
  = \frac{\tau^{\Hp} - T^{\Hp}}{\Hp} + \Oo(s).
$$
As a function of~$s$, $\Sigma_s^{\tf_2,\tau}$ is continuously differentiable.
Because we are integrating over the compact $[0,t]$, we can integrate term by term, so that
\begin{align*}
\log\Cf_t^{\tau}(\tf_1, \tf_2) 
 & = 
\int_{0}^{t}\left(\Sigma_s^{\tf_2,\tau} - \Sigma_s^{\tf_1,\tau}\right)
\Sigma_s^{\tf_2,\tau}\D s\\
 & = 
\int_{0}^{t}
\left\{\left(\frac{\tau^{\Hp} - \tf_{2}^{\Hp}}{\Hp} - \frac{\tau^{\Hp} - \tf_{1}^{\Hp}}{\Hp}+ \Oo(s)\right)\left(\frac{\tau^{\Hp} - \tf_{2}^{\Hp}}{\Hp}
 + \Oo(s)\right)\right\}\D s\\
 & = 
\int_{0}^{t}
\left\{\left(\frac{\tf_{1}^{\Hp} - \tf_{2}^{\Hp}}{\Hp} + \Oo(s)\right)\left(\frac{\tau^{\Hp} - \tf_{2}^{\Hp}}{\Hp}
 + \Oo(s)\right)\right\}\D s\\
 & = 
\frac{\tf_{1}^{\Hp} - \tf_{2}^{\Hp}}{\Hp}
\frac{\tau^{\Hp} - \tf_{2}^{\Hp}}{\Hp}t
 + \Oo(t^2),
\end{align*}
where we can check by direct computations that the term $\Oo(t^2)$ is indeed non null.      
\end{proof}

\subsection{Extension to smooth Gaussian Volterra semimartingale drivers}

Let now~$\Wf$ in~\eqref{eq:ShortRatedWH} be a Gaussian Volterra process with a smooth kernel of the form
$$
\Wf_t = \int_{0}^{t}K(t,u)\D W_u,
$$
for some standard Brownian motion~$W$.
Assuming that~$K$ is a convolution kernel absolutely continuous with square integrable derivative,
it follows by~\cite{basse2010path} that~$\Wf$ is a Gaussian semimartingale (yet not necessarily a martingale) with the decomposition
$$
\Wf_t = \int_{0}^{t}K(u,u)\D W_u + \int_{0}^{t}\left(\int_{0}^{u}\partial_1 K(u,s)\D W_s\right)\D u
 =: \int_{0}^{t}K(u,u)\D W_u + A(t),
$$
where~$A$ is a process of bounded variation satisfying
$\D A(t) = A'(t)\D t = \left(\int_{0}^{t}\partial_1 K(t,s)\D W_s\right)\D t$
and hence the It\^o differential of~$\Wf_t$ reads
$\D \Wf_t = K(t,t)\D W_t + A'(t)\D t$,
and its quadratic variation is
$\D \langle \Wf, \Wf\rangle_t = \int_{0}^{t}K(u,u)^2 \D u$.
The short rate process~\eqref{eq:ShortRatedWH} therefore reads
$$
r_t = \theta(t) + \int_{0}^{t}\varphi(t-u)\D \Wf_u
 = \theta(t) + \int_{0}^{t}\varphi(t-u)\left(K(u,u)\D W_u + A'(u)\D u\right)
 = \widetilde{\theta}_t + \int_{0}^{t}\widetilde{\varphi}(t,u)K(u,u)\D W_u,
$$
where $\displaystyle \widetilde{\theta}_t:=\theta + \int_{0}^{t}\varphi(t-u)A'(u)\D u$
and $\displaystyle \widetilde{\varphi}(t,u) := \varphi(t-u)K(u,u)$. 
If~$\widetilde{\varphi}$ satisfies Assumption~\ref{assu:Kernel}, then the analysis above still holds.



\subsubsection{Comments on the Bond process}
Let $R_{t,T}  := \int_{t}^{T} r_s \D s$ be the integrated short rate process
and $B_{t,T} := \E^{-R_{t,T}}$ the bond price process on $[0,T]$.
\begin{lemma}\label{lem:smoothEx}
The process $(B_{t,T})_{t\in [0,T]}$ satisfies $B_{T,T}=1$ and, for $t \in [0,T)$,
$$
\frac{\D B_{t,T}}{B_{t,T}}
 = r_t \D t
 = \left(\theta(t) + \int_{0}^{t}\varphi(t-u) A'(u) \D u
+ \int_{0}^{t} \varphi(t-u)K(u,u)\D W_u\right)\D t.
$$
\end{lemma}

\begin{proof}
For any $t\in [0,T)$, we can write
\begin{align*}
r_t & = \theta(t) + \int_{0}^{t}\varphi(t-u) \D\left(\int_{0}^{u}K(s,s)\D W_s + A(u)\right)
 = \theta(t) + \int_{0}^{t}\varphi(t-u) A'(u) \D u+  \int_{0}^{t}\varphi(t-u) K(u,u)\D W_u.
\end{align*}
and therefore 
\begin{equation}\label{eq:SDE_RtT}
\D R_{t,T} = -r_t \D t
 = -\left(\theta(t) + \int_{0}^{t}\varphi(t-u) A'(u) \D u
 +\int_{0}^{t}\varphi(t-u) K(u,u)\D W_u\right)\D t.
\end{equation}
It\^o's formula~\cite[Theorem 4]{alos2001stochastic} then yields
\begin{align*}
B_{T,T} & = B_{t,T} - \int_{t}^{T} B_{s,T}\D R_{s,T} + \half\int_{t}^{T} B_{s,T}\D \langle R, R\rangle_{s,T} \\
 & = B_{t,T} + \int_{t}^{T} B_{s,T}\left\{
\left(\theta(s) + \int_{0}^{s}\varphi(s,u) A'(u) \D u\right)
+ \int_{0}^{s} \varphi(s,u)K(u,u)\D W_u\right\}\D s.
\end{align*}
so that, since $B_{T,T}=1$,
the lemma follows from
\begin{align*}
\D B_{t,T} & = -\D\left(\int_{t}^{T} B_{s,T}\left\{
\left(\theta(s) + \int_{0}^{s}\varphi(s,u) A'(u) \D u\right)
+ \int_{0}^{s} \varphi(s,u)K(u,u)\D W_u\right\}\D s\right)\\
& = B_{t,T}\left\{
\left(\theta(t) + \int_{0}^{t}\varphi(t-u) A'(u) \D u\right)
+ \int_{0}^{t} \varphi(t-u)K(u,u)\D W_u\right\}\D t.
\end{align*}
\end{proof}

\begin{remark}
We can also write~$R_{t,T}$ in integral form as follows, using stochastic Fubini:
\begin{align*}
R_{t,T} & = \int_{t}^{T}\left[\theta(s) + \int_{0}^{s}\varphi(s,u) A'(u) \D u +  \int_{0}^{s}\varphi(s,u) K(u,u)\D W_u\right]\D s\\
 & = \Theta_{t,T} + \int_{t}^{T}\left(\int_{0}^{s}\varphi(s,u) A'(u) \D u \right) \D s + \int_{t}^{T}\left( \int_{0}^{s}\varphi(s,u)K(u,u)\D W_u\right)\D s\\
 & = \Theta_{t,T} + \int_{0}^{t}\left(\int_{t}^{T}\varphi(s,u)\D s \right)  A'(u) \D u + \int_{0}^{t}\left( \int_{t}^{T}\varphi(s,u)\D s\right) K(u,u)\D W_u\\
 &\qquad\quad + \int_{t}^{T}\left(\int_{u}^{T}\varphi(s,u)\D s \right)  A'(u) \D u + \int_{t}^{T}\left( \int_{u}^{T}\varphi(s,u)\D s\right) K(u,u)\D W_u\\
 & = \Theta_{t,T} + \int_{0}^{t}\Phi_{t}(u) A'(u) \D u + \int_{0}^{t}\Phi^K_{t}(u)\D W_u + \int_{t}^{T}\Phi_u(u) A'(u) \D u + \int_{t}^{T}\Phi^K_u(u)\D W_u,
\end{align*}
with 
$\Phi_{t}(u) := \displaystyle\int_{t}^{T}\varphi(s,u)\D s$ and
$\Phi^K_t(u) := \Phi_{t}(u)K(u,u)$.
As a consistency check, we have
\begin{align*}
\D R_{t,T} & = -\theta(t)\D t + \Phi_{t}(t) A'(t) \D t + \Phi^K_{t}(t)\D W_t - \Phi_t(t) A'(t) \D t -\Phi^K_t(t)\D W_t
+ \int_{0}^{t}\partial_t\Phi_{t}(u) A'(u) \D u\D t + \int_{0}^{t}\partial_t\Phi^K_{t}(u)\D W_u\D t\\
& = \left(-\theta(t) + \Phi_{t}(t) A'(t) - \Phi_t(t) A'(t) + \int_{0}^{t}\partial_t\Phi_{t}(u) A'(u) \D u + \int_{0}^{t}\partial_t\Phi^K_{t}(u)\D W_u\right)\D t 
+ \left(\Phi^K_{t}(t) -\Phi^K_t(t)\right)\D W_t\\
& = \left(-\theta(t) + \Phi_{t}(t) A'(t) - \Phi_t(t) A'(t) + \int_{0}^{t}\partial_t\Phi_{t}(u) A'(u) \D u \right)\D t 
+ \int_{0}^{t} \partial_t\Phi^K_{t}(u)\D W_u\D t\\
& = \left(-\theta(t) + \int_{0}^{t}\partial_t\Phi_{t}(u) A'(u) \D u \right)\D t 
+ \int_{0}^{t} \partial_t\Phi^K_{t}(u)\D W_u\D t\\
& = -\left(\theta(t) + \int_{0}^{t}\varphi(t-u) A'(u) \D u \right)\D t 
- \int_{0}^{t} \varphi(t-u)K(u,u)\D W_u\D t,
\end{align*}
which corresponds precisely to~\eqref{eq:SDE_RtT}.
\end{remark}

\subsubsection{Specific example}\label{sec:NotBrownian}
Consider the kernel
$K(t,s) = \E^{-\beta (t-s)}$ with $\beta>0$, 
so that $K(t,t) = 1$ and $\partial_{t}K(t,s) = -\beta K(t,s)$.
In this case, 
$A'(t) = \int_{0}^{t}\partial_{t}K(t,s)\D W_s = -\beta \Wf_t$,
so that 
$\D\Wf_t = \D W_t -\beta \Wf_t\D t$, 
which is an Ornstein-Uhlenbeck process,
with covariance, for all $s, t\geq 0$,
$$
\EE[\Wf_s \Wf_t]  = \EE\left[\int_{0}^{s}K(s,u)\D W_u \cdot\int_{0}^{t}K(t,u)\D W_u\right] = \int_{0}^{s}K(s,u)K(t,u)\D u = \frac{\E^{-\beta|t-s|} - \E^{-\beta(s+t)}}{2\beta}.
$$

The short rate dynamics in~\eqref{eq:ShortRatedWH} then reads
$$
r_t = \theta(t) + \int_{0}^{t}\varphi(t,u)\D \Wf_u
 = \theta(t) + \int_{0}^{t}\varphi(t,u)\Big(\D W_u -\beta \Wf_u\D u\Big)
 = \widetilde{\theta}(t) + \int_{0}^{t}\varphi(t,u)\D W_u,
 $$
with
$\widetilde{\theta}(t) = \theta(t) -\beta\int_{0}^{t}\varphi(t-u)\Wf_u\D u$,
and the zero-coupon bond dynamics
(Proposition~\ref{prop:ZeroCoupon}) reads
$$
P_{t,T} 
 = 
\exp\left\{-\widetilde{\Theta}_{t,T} + \half\int_{t}^{T}\XiT(u)^2 \D u + \int_{0}^{t}\XiT(t,u)\D W_u\right\},
$$
with 
$$
\widetilde{\Theta}_{t,T}
:= \int_{t}^{T}\widetilde{\theta}_s \D s
= \int_{t}^{T}\Big(\theta(s) -\beta\int_{0}^{s}\varphi(s-u)\Wf_u\D u\Big)\D s
= \Theta_{t,T}
- \beta \int_{t}^{T}\int_{0}^{s}\varphi(s-u)\Wf_u\D u\, \D s.
$$
Applying stochastic Fubini, we then obtain
$$\widetilde{\Theta}_{t,T}
  = \Theta_{t,T}
- \beta \int_{0}^{T}\int_{t+u(1-t/T)}^{T}\varphi(s-u)\D s\,\Wf_u\D u
  = \Theta_{t,T}
- \beta \int_{0}^{T}
\left(\Phi(T-u) - \Phi\left(t-\frac{t}{T}\right)\right)\Wf_u\D u.
$$
We note that the convexity adjustment in Proposition~\ref{prop:Convexity} is only affected by a different weighting scheme in the integral given by the function~$\gamma'_{\Wf}$.
In our case, from the covariance computation above,
$\gamma_{\Wf}(t) = \frac{1}{2\beta}(1-\E^{-2\beta t})$,
and therefore
$\gamma'_{\Wf}(t) = \E^{-2\beta t}$.
\section{Pricing OIS products and options}
\label{sec:Products}

\subsection{Simple compounded rate}

Using Proposition~\ref{prop:ZeroCoupon}, we can compute several OIS products and options
Consider the simple compounded rate
\begin{equation}\label{eq:Rate}
r^S(t_0,T) := \frac{1}{\Df(t_0,T)}\left(\prod_{i=0}^{n-1}\frac{1}{P_{t_i,t_{i+1}}} - 1\right),
\end{equation}
where~$\Df(t_0,T)$  is the day count fraction and $n$ the number of business days in the period
$[t_0,t_n]$.
The following then holds directly:
$$
r^S(t_0^R,T) = \frac{1}{\Df(t_0,T)}\left(\prod_{i=0}^{n-1}
\exp\left\{\Theta_{t_{i}^R,t_{i+1}^R} - \half\int_{t_{i}^R}^{t_{i+1}^R}\Xi(u,u)^2 \D u - \left(\Xi(t_{i}^R,\cdot)\circ \Wf\right)_{t_{i}^R}\right\}
 - 1\right),
$$
where the superscript~$^R$ refers to reset dates;
we use the superscript~$^A$ to refer to accrual dates below.

\subsection{Compounded rate cashflows with payment delay}
The present value at time zero of a compounded rate cashflow is given by
\begin{align*}
\mathrm{PV}_{\mathrm{flow}}
 & = P_{0,T_{p}}\Df(t_0^A, t_n^A)\EE^{\QQ^{T_{p}}}\left[r^{S}\right]\\
 & = P_{0,T_{p}}\Df(t_0^A, t_n^A)
\EE^{\QQ^{T_{p}}}\left[\frac{1}{\Df(t_0^A, t_n^A)}\left\{\prod_{i=0}^{n-1}\left(1+\frac{\Df(t_i^A,t_{i+1}^A)}{\Df(t_i^R,t_{i+1}^R)}\left(\frac{P_{t,t_i^R}}{P_{t,t_{i+1}^R}}-1\right)\right)-1\right\}\right],
\end{align*}
where $r^{S}$ denotes the compounded RFR rate.
In the case where there is no reset delays, namely $t_i^R=t_i^A$ for all $i=0,\ldots, n$, then
\begin{align*}
\mathrm{PV}_{\mathrm{flow}}
 = P_{0,T_{p}}
\EE^{\QQ^{T_{p}}}\left[\prod_{i=0}^{n-1}\left(\frac{P_{t,t_i^R}}{P_{t,t_{i+1}^R}}\right)-1\right]
 & = P_{0,T_{p}}
\EE^{\QQ^{T_{p}}}\left[\frac{P_{t,t_0^R}}{P_{t,t_{n}^R}} - 1\right]\\
&  = P_{0,T_{p}}
\left(\frac{P_{0,t_0^R}}{P_{0,t_{n}^R}}\Cf_{t}^{T_{p}}(t_0^R,t_n^R) - 1\right)\\
&  = P_{0,T_{p}}
\left(\frac{P_{0,T_{RS}}}{P_{0,T_{RE}}}\Cf_{t}^{T_{p}}(T_{RS},T_{RE}) - 1\right),
\end{align*}
where $t_0^R = T_{RS}$ and $t_n^R = T_{RE}$,
using the convexity adjustment formula given in Proposition~\ref{prop:Convexity}.

\subsection{Compounded rate cashflows with reset delay}
Assuming now that $t_i^R \ne t_i^A$, we can write
$r^S_t = \widetilde{r}_t^S + r_{t}^{S, adj}$,
from~\eqref{eq:Rate},
where
$$
\widetilde{r}^S_t := \frac{1}{\Df(t_0^R,t_n^R)}\left(\frac{P_{t,T_{RS}}}{P_{t, T_{RE}}} - 1\right),
$$
and $r^{S, adj}_t$ is implied from the decomposition above.
Therefore
\begin{align*}
\mathrm{PV}_{\mathrm{flow}}
 & = P_{0,T_{p}}\Df(t_0^A, t_n^A)\EE^{\QQ^{T_{p}}}\left[r^{S}_t\right]\\
 & = P_{0,T_{p}}\Df(t_0^A, t_n^A)\EE^{\QQ^{T_{p}}}\left[\widetilde{r}_t^S + r_{t}^{S, adj}\right]\\
 & = P_{0,T_{p}}\Df(t_0^A, t_n^A)\EE^{\QQ^{T_{p}}}\left[\frac{1}{\Df(t_0^R,t_n^R)}\left(\frac{P_{t,T_{RS}}}{P_{t, T_{RE}}} - 1\right) + r_{t}^{S, adj}\right]\\
 & = P_{0,T_{p}}\Df(t_0^A, t_n^A)
 \left\{\frac{1}{\Df(t_0^R,t_n^R)}\left(\frac{P_{0,T_{RS}}}{P_{0,T_{RE}}}\Cf_{t}^{T_{p}}(T_{RS},T_{RE}) - 1\right) + \EE^{\QQ^{T_{p}}}\left[r_{t}^{S, adj}\right]\right\}\\
 & = P_{0,T_{p}}\frac{\Df(t_0^A, t_n^A)}{\Df(t_0^R,t_n^R)}
 \left\{\frac{P_{0,T_{RS}}}{P_{0,T_{RE}}}\Cf_{t}^{T_{p}}(T_{RS},T_{RE}) - 1
  + \Df(t_0^R,t_n^R)\EE^{\QQ^{T_{p}}}\left[r_{t}^{S, adj}\right]\right\}.
\end{align*}
Assume now that $\EE^{\QQ^{T_{p}}}\left[r_{t}^{S, adj}\right] = r_{0}^{S, adj}$,
so that we can simplify the above as 
\begin{align*}
\mathrm{PV}_{\mathrm{flow}}
 & = P_{0,T_{p}}\frac{\Df(t_0^A, t_n^A)}{\Df(t_0^R,t_n^R)}
 \left\{\frac{P_{0,T_{RS}}}{P_{0,T_{RE}}}\Cf_{t}^{T_{p}}(T_{RS},T_{RE}) - 1
  + \Df(t_0^R,t_n^R)r_{0}^{S, adj}\right\}\\
 & = P_{0,T_{p}}\frac{\Df(t_0^A, t_n^A)}{\Df(t_0^R,t_n^R)}
 \left\{\frac{P_{0,T_{RS}}}{P_{0,T_{RE}}}\Cf_{t}^{T_{p}}(T_{RS},T_{RE}) - 1
  + \Df(t_0^R,t_n^R)\left(r_{0}^{S} - \widetilde{r}_{0}^{S}\right)\right\}\\
 & = P_{0,T_{p}}\frac{\Df(t_0^A, t_n^A)}{\Df(t_0^R,t_n^R)}
 \left\{\frac{P_{0,T_{RS}}}{P_{0,T_{RE}}}\Cf_{t}^{T_{p}}(T_{RS},T_{RE}) - 1
  + \Df(t_0^R,t_n^R)\left(r_{0}^{S} - \frac{1}{\Df(t_0^R,t_n^R)}\left(\frac{P_{0,T_{RS}}}{P_{0, T_{RE}}} - 1\right)\right)\right\}\\
 & = P_{0,T_{p}}\frac{\Df(t_0^A, t_n^A)}{\Df(t_0^R,t_n^R)}
 \left\{\frac{P_{0,T_{RS}}}{P_{0,T_{RE}}}\left(\Cf_{t}^{T_{p}}(T_{RS},T_{RE}) -1 \right)
  + \Df(t_0^R,t_n^R) r_{0}^{S}\right\}.
\end{align*}

\section{Numerics}
\label{sec:Numerics}
\subsection{Zero-coupon dynamics}
In Figure~\ref{fig:ZeroCouponDynamics_Exp}, we analyse the impact of the parameter, $\alpha$ in the Exponential kernel case (Section~\ref{sec:Ex_ExpKernel})
and~$H$ in the Riemann-Liouville case (Section~\ref{sec:Ex_RLKernel}),
on the dynamics of the zero-coupon bond over a time span $[0,1]$ and considering a constant curve $\theta(\cdot) = 6\%$.
In order to compare them properly, the underlying Brownian path is the same for all kernels.
Unsurprisingly, we observe that the Riemann-Liouville case creates a lot more variance of the dynamics.

\begin{figure}[H]
\centering
\includegraphics[scale=0.45]{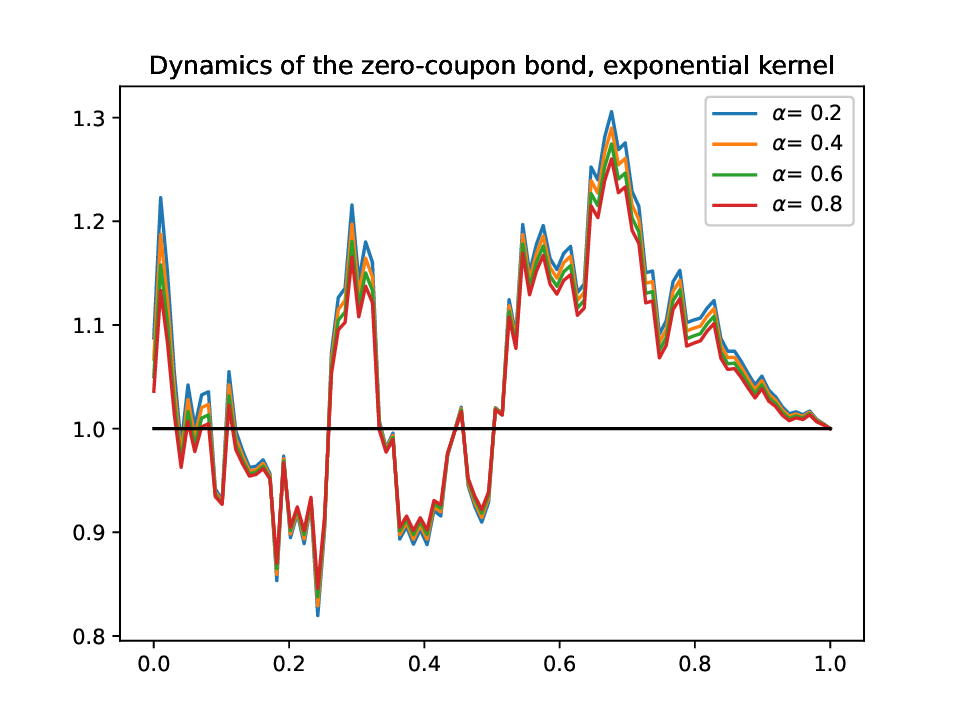}
\includegraphics[scale=0.45]{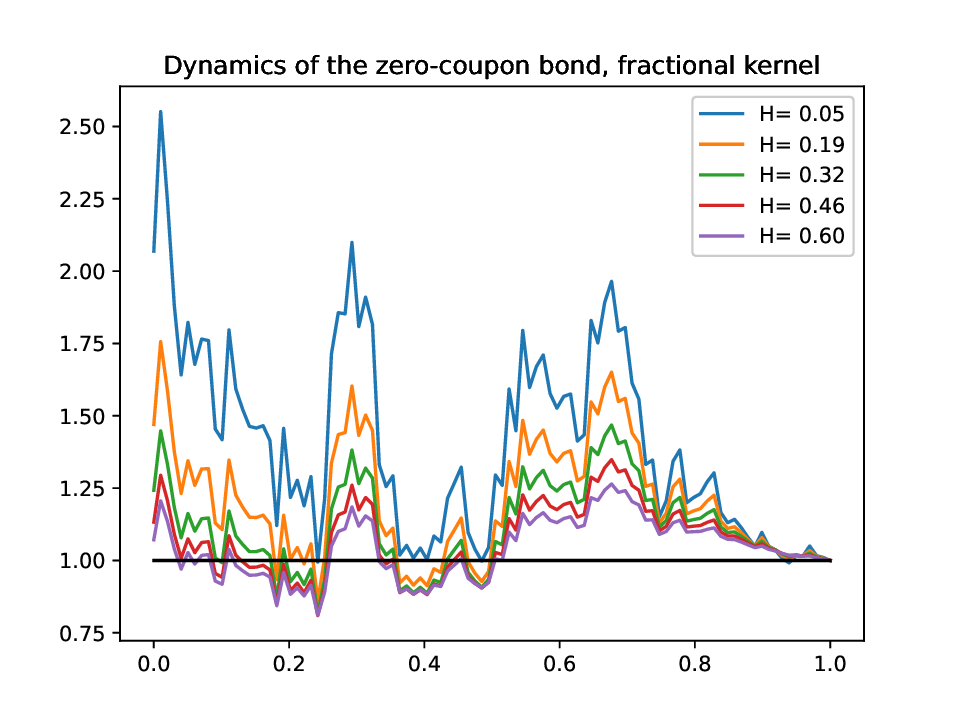}
\caption{Dynamics of the zero-coupon bond in the Exponential (left) and the Riemann-Liouville (right) kernel case.}
\label{fig:ZeroCouponDynamics_Exp}
\end{figure}

\subsection{Impact of the roughness on convexity}

We compare in Figure~\ref{fig:Impact_ExpKernel} the impact of the (roughness of the) kernel on the convexity adjustment.
We consider a constant curve 
$\theta(\cdot) = 6\%$ and
$(t, \tf_1, \tf_2, \tau) = (1, 2, 3, 2)$.
As~$\alpha$ tends to zero (exponential kernel case) and as~$H$ tends to~$\half$ (Riemann-Liouville case), 
the convexity adjustments converge to the same value (as expected), 
approximately equal to~$2.718$.
In Figure~\ref{fig:NoBrownian}, 
we consider Example~\ref{sec:NotBrownian},
shifting away from a standard Brownian driver.

\begin{figure}[H]
\centering
\includegraphics[scale=0.35]{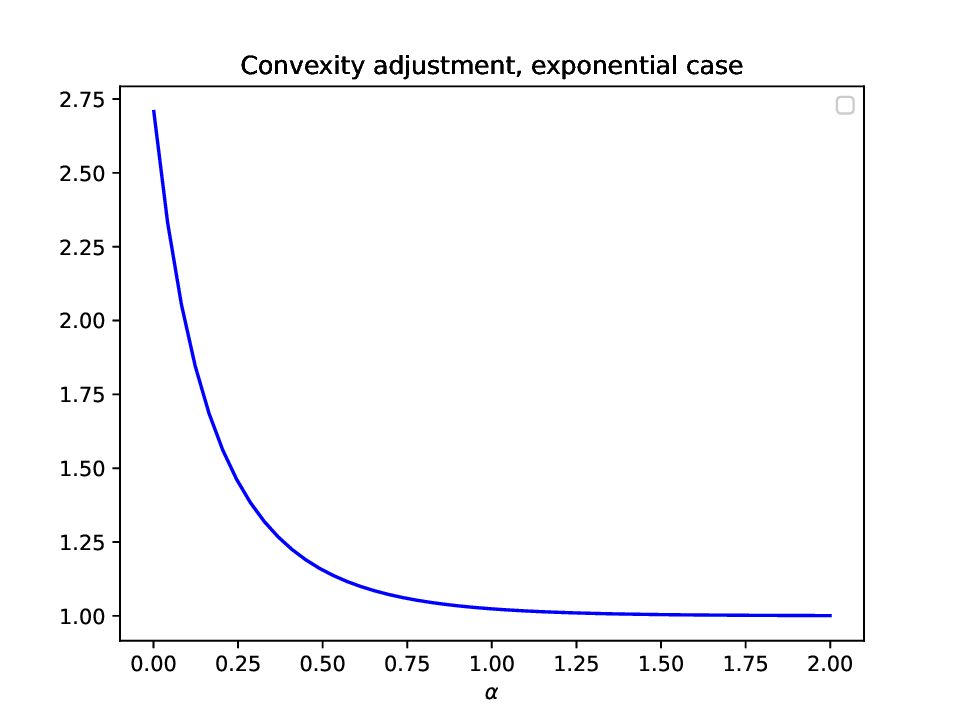}
\includegraphics[scale=0.35]{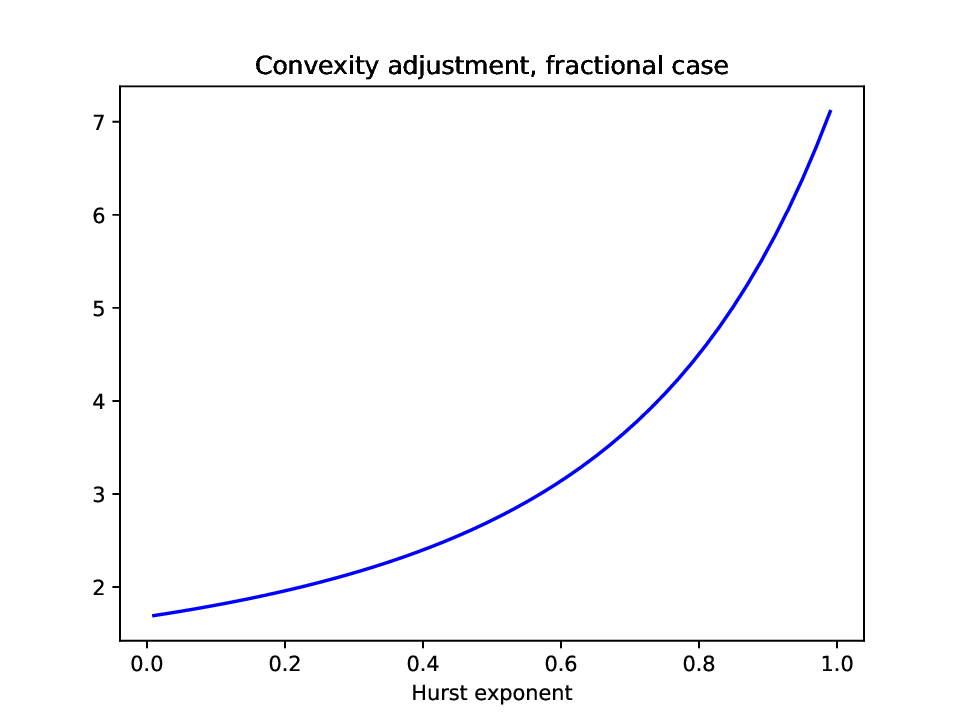}
\caption{Left: Impact of the exponential factor~$\alpha$ on the convexity for the Exponential kernel from Section~\ref{sec:Ex_ExpKernel}.
Right: Impact of the Hurst exponent~$H$ on the convexity for the power-law kernel from Section~\ref{sec:Ex_RLKernel}}
\label{fig:Impact_ExpKernel}
\end{figure}

\begin{figure}[H]
\centering
\includegraphics[scale=0.45]{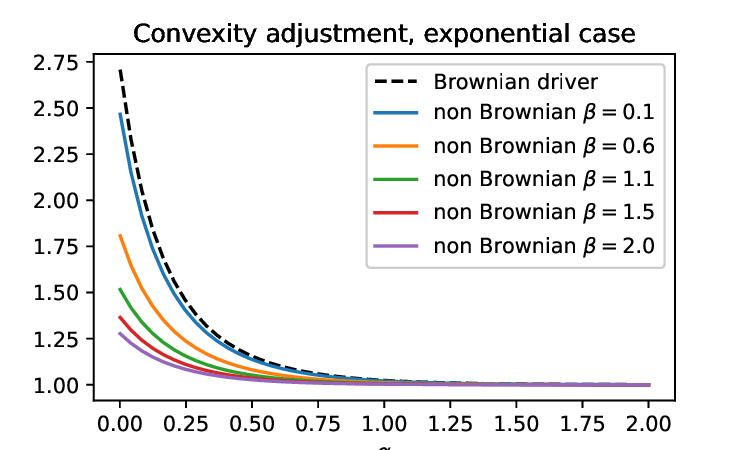}
\includegraphics[scale=0.45]{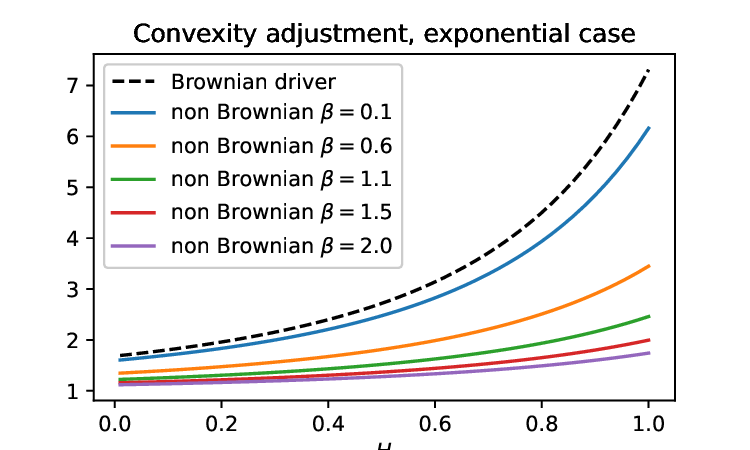}
\caption{Left: Impact of the exponential factor~$\alpha$ on the convexity for the Exponential kernel with standard Brownian motion (black dashed) and with OU driver with different~$\beta$ parameters.
Right: Same but with the power-law kernel.}
\label{fig:NoBrownian}
\end{figure}

\section*{Disclosure of interest}
We confirm that there are no relevant financial or non-financial competing interests to report.

\newpage
\bibliography{References}

\begin{thebibliography}{10}

\bibitem{alos2001stochastic}
{\sc E.~Al\`os, O.~Mazet, and D.~Nualart}, {\em Stochastic calculus with
  respect to {G}aussian processes}, The Annals of Probability, 29 (2001),
  pp.~766--801.

\bibitem{andersen2010interest}
{\sc L.~Andersen and V.~Piterbarg}, {\em Interest Rate Modeling}, Atlantic
  Financial Press, 2010.

\bibitem{basse2010path}
{\sc A.~Basse-O’Connor and S.-E. Graversen}, {\em Path and semimartingale
  properties of chaos processes}, Stochastic Processes and their Applications,
  120 (2010), pp.~522--540.

\bibitem{bayer2016pricing}
{\sc C.~Bayer, P.~Friz, and J.~Gatheral}, {\em Pricing under rough volatility},
  Quantitative Finance, 16 (2016), pp.~887--904.

\bibitem{bonesini2023rough}
{\sc O.~Bonesini, A.~Jacquier, and A.~Pannier}, {\em Rough volatility,
  path-dependent {PDEs} and weak rates of convergence}.
\newblock \href{https://arxiv.org/abs/2304.03042}{arXiv:2304.03042}, 2023.

\bibitem{brigo2006interest}
{\sc D.~Brigo and F.~Mercurio}, {\em Interest rate models-theory and practice:
  with smile, inflation and credit}, vol.~2, Springer, 2006.

\bibitem{brotherton1993yield}
{\sc R.~Brotherton-Ratcliffe and B.~Iben}, {\em Yield curve applications of
  swap products}, Schwartz, RJ, Smith, CW, Advanced Strategies in Financial
  Risk Management, New York Institute of Finance, New York,  (1993),
  pp.~400--450.

\bibitem{el2018microstructural}
{\sc O.~El~Euch, M.~Fukasawa, and M.~Rosenbaum}, {\em The microstructural
  foundations of leverage effect and rough volatility}, Finance and
  Stochastics, 22 (2018), pp.~241--280.

\bibitem{flesaker1993arbitrage}
{\sc B.~Flesaker}, {\em Arbitrage free pricing of interest rate futures and
  forward contracts}, The Journal of Futures Markets (1986-1998), 13 (1993),
  p.~77.

\bibitem{fukasawa2021volatility}
{\sc M.~Fukasawa}, {\em Volatility has to be rough}, Quantitative Finance, 21
  (2021), pp.~1--8.

\bibitem{garcia2023convexity}
{\sc D.~Garc{\'\i}a-Lorite and R.~Merino}, {\em Convexity adjustments \`a la
  {M}alliavin}.
\newblock \href{https://arxiv.org/abs/2304.13402}{arXiv:2304.13402}, 2023.

\bibitem{gatheral2022volatility}
{\sc J.~Gatheral, T.~Jaisson, and M.~Rosenbaum}, {\em Volatility is rough},
  Quantitative Finance, 18 (2018), pp.~933--949.

\bibitem{hull1990pricing}
{\sc J.~Hull and A.~White}, {\em Pricing interest-rate-derivative securities},
  The review of financial studies, 3 (1990), pp.~573--592.

\bibitem{jacquier2021rough}
{\sc A.~Jacquier, A.~Muguruza, and A.~Pannier}, {\em Rough multifactor
  volatility for {SPX} and {VIX} options}.
\newblock \href{https://arxiv.org/abs/2112.14310}{arXiv:2112.14310}, 2021.

\bibitem{jacquier2023deep}
{\sc A.~Jacquier and M.~Oumgari}, {\em Deep curve-dependent {PDEs} for affine
  rough volatility}, SIAM Journal on Financial Mathematics, 14 (2023),
  pp.~353--382.

\bibitem{kirikos1997convexity}
{\sc G.~Kirikos and D.~Novak}, {\em Convexity conundrums: Presenting a
  treatment of swap convexity in the hall-white framework}, Risk Magazine, 10
  (1997), pp.~60--61.

\bibitem{oksendal2003stochastic}
{\sc B.~{\O}ksendal}, {\em Stochastic differential equations}, Springer, 2003.

\bibitem{pelsser2003mathematical}
{\sc A.~Pelsser}, {\em Mathematical foundation of convexity correction},
  Quantitative Finance, 3 (2003), pp.~59--65.

\bibitem{ritchken1993averaging}
{\sc P.~Ritchken and L.~Sankarasubramanian}, {\em Averaging and deferred
  payment yield agreements}, The Journal of Futures Markets (1986-1998), 13
  (1993), p.~23.

\bibitem{russo2006bifractional}
{\sc F.~Russo and C.~A. Tudor}, {\em On bifractional {B}rownian motion},
  Stochastic Processes and their Applications, 116 (2006), pp.~830--856.

\bibitem{vasicek1977equilibrium}
{\sc O.~Vasicek}, {\em An equilibrium characterization of the term structure},
  Journal of {F}inancial {E}conomics, 5 (1977), pp.~177--188.

\end{thebibliography}
\bibliographystyle{siam}


\end{document}